\documentclass[12pt,draftclsnofoot,onecolumn]{IEEEtran}                                                        
\IEEEoverridecommandlockouts                              
\def\BibTeX{{\rm B\kern-.05em{\sc i\kern-.025em b}\kern-.08em
		T\kern-.1667em\lower.7ex\hbox{E}\kern-.125emX}}
	
\usepackage[utf8]{inputenc}

\usepackage{graphics} 
\usepackage{epsfig} 
\usepackage{graphicx}
\usepackage[caption=false]{subfig}

\usepackage{multirow}
\usepackage{amsmath}
\usepackage{amssymb}
\usepackage{amsfonts}
\usepackage{amsthm}
\usepackage{color}
\usepackage{booktabs}
\usepackage{float}
\usepackage{bm}
\usepackage{mathtools}
\usepackage{soulutf8} 
\usepackage{xcolor}
\usepackage{array}
\usepackage{soul}
\usepackage{array}
\usepackage{varwidth}

\theoremstyle{remark}
\newtheorem{remark}{Remark}
 
\usepackage{algorithmicx}
\usepackage{algorithm}
\usepackage{algpseudocode}

\usepackage[caption=false]{subfig}

\newcolumntype{P}[1]{>{\centering\arraybackslash}p{#1}}

\algnewcommand\algorithmicinput{\textbf{Input:}}
\algnewcommand\Input{\item[\algorithmicinput]}

\algnewcommand\SET[2]{\item\algorithmicset\ #1 \algorithmicto\ #2}

\makeatletter
\def\BState{\State\hskip-\ALG@thistlm}
\makeatother

\theoremstyle{plain}
\newtheorem{theorem}{Theorem}
\newtheorem{corollary}{Corollary}

\theoremstyle{definition}
\newtheorem{definition}{Definition}

\definecolor{orange}{RGB}{255,127,0}
\definecolor{gold}{rgb}{0.85,.66,0}

\usepackage{threeparttable}

\title{Randomized Kaczmarz Algorithm for Massive MIMO Systems with Channel Estimation and Spatial Correlation}

\author{
	\IEEEauthorblockN{Victor Croisfelt Rodrigues\IEEEauthorrefmark{1}, Jos\'{e} Carlos Marinello Filho\IEEEauthorrefmark{1}, and Taufik Abr\~{a}o\IEEEauthorrefmark{1}}\\
	\IEEEauthorblockA{\IEEEauthorrefmark{1} Department of Electrical Engineering (DEEL), State University of Londrina (UEL), Londrina, Brazil\\
	E-mail: victorcroisfelt@gmail.com, zecarlos.ee@gmail.com, and taufik@uel.br}
}

\begin{document}
\maketitle
\vspace{-10mm}

\begin{abstract}
	To exploit the benefits of massive multiple-input multiple-output (M-MIMO) technology in scenarios where base stations (BSs) need to be {cheap} and equipped with simple hardware, the computational complexity of classical signal processing schemes for spatial multiplexing {of users shall} be reduced. This calls for suboptimal designs that perform well the combining/precoding steps and simultaneously achieve low computational complexities. An approach based on the iterative Kaczmarz algorithm (KA) has been recently investigated, assuring well {execution} without the knowledge of {second order moments of the wireless channels in the BS}, and with {easiness}, since no {tuning parameters,} besides the number of iterations, are required. In fact, the randomized version of KA (rKA) {has been used in this context} due to global convergence properties. Herein, modifications are proposed on {this first} rKA-based attempt, {aiming to improve its} performance--complexity trade-off solution for {M-MIMO systems}. We observe that long-term channel effects degrade the rate of convergence of the rKA-based schemes. This issue is {then tackled herein} by means of a hybrid {rKA} initialization proposal that lands within the region of convexity {of the algorithm} and assures fairness to the {communication system}. {The effectiveness of our proposal is illustrated through numerical results which bring} more realistic system conditions in terms of channel estimation and spatial correlation {than those used so far}. {We also characterize the computational complexity of the proposed rKA scheme, deriving upper bounds for the number of iterations. A case study focused on a dense urban application scenario is used to gather new insights on the feasibility of the proposed scheme to cope with the inserted BS constraints.}
\end{abstract}


\section{Introduction}
{Space-division multiple access is the main concept behind the massive multiple-input multiple-output (M-MIMO) technology which allows all concurrent user equipment (UE) to exploit the same time-frequency resources for communication.} In each coherence block, therefore, a base station (BS) equipped with a {large array of antennas needs to estimate the responses of the wireless channels that connect its antennas to the UEs requesting connection}, so that it can spatially differentiate these UEs and execute the signal processing techniques {required for dependable exchange of information}. This knowledge, {often called channel state information (CSI), is canonically acquired} through a pilot training phase, which turns out to be generally feasible by the application of a time-division duplex (TDD) architecture. Although non-linear M-MIMO signal processing techniques are typically credited as optimum implementation strategies, it is well-known that linear schemes achieve an adequate spatial resolution of UEs{, i.e., they have} compelling performance--(computational) complexity results from the point of view of implementation\cite{Marzetta2010}. This is {actually} a direct consequence of the use of massive antenna arrays that, basically, suppresses the interference among UEs \cite{Marzetta2016,Bjornson2017}. In view of the fact that the CSI is proportional to the number of both BS and UEs' antennas adopted in the system, the computation of classical combining/precoding schemes{, unfortunately, still demands} substantially large computational complexities {when considering a BS with limited processing power}. For this reason, {in order to make} BSs cheaper {in the sense of processing hardware}, it is desirable to relax further the computational complexities related to canonical linear schemes, or simply canonical schemes.

Particularly, a good strategy to reduce the computational complexity of signal processing schemes has been the search for suboptimal techniques. These called "relaxed" schemes are envisaged to lessen the complexity of the canonical ones and, simultaneously, achieve similar capacity results to them. The recent literature contains different proposals to obtain the aforementioned goals, such as the {truncated polynomial expansion (TPE) method \cite{Mueller2014}. Despite the good results reported in \cite{Mueller2014}, the advantages of TPE-based schemes} extremely rely on the statistical knowledge of the second order moment of the channel responses, that is, the information of the channel covariance matrices must be known at the receiver {to perform such approach}. {However}, as can be seen from \cite{Bjornson2017a}, {for example}, estimating the second order is undoubtedly challenging and computationally prohibitive to be obtained, even more in scenarios {with hardware limitation}. {This comes to be} especially {true} when {using large antenna arrays}, rendering it unattainable in the context {of low complexity BSs} under consideration, without even need to take into account the severe storage requirements {brought by them}.

In contrast to TPE, the authors of \cite{Boroujerdi2018b} proposed the application of the Kaczmarz algorithm (KA) over the combining/precoding problems of M-MIMO. A remarkable advantage of KA-based methods over TPE or even other schemes lies in the fact that the former does not require the knowledge of covariance matrices to achieve suitable performance--complexity results. The procedure was initially proposed by Kaczmarz in \cite{Kaczmarz1937} as an iterative technique for solving determined and overdetermined (OD) set of linear equations (SLE). With the increasing popularity of stochastic gradient techniques and machine learning, KA has been revived \cite{Needell2014} and applied to other problems, such as solving quadratic equations \cite{Chi2016}. Recently, a randomized version of the KA (rKA) to solve consistent OD SLEs was introduced and analyzed in \cite{Strohmer2006}. This random approach of the KA obtains an expected exponential rate of convergence and is therefore much more reliable, unlike the case of classical KA convergence that depends tremendously on the way that the equations are arranged in {an SLE to be solved}. For M-MIMO, the authors of \cite{Boroujerdi2018b} derived a mathematical framework to examine the performance of rKA applied to emulate some linear signal processing techniques. The rKA-based schemes of \cite{Boroujerdi2018b} consider the resolution of the signal estimation problems, i.e., they are used to estimate in the receiver/transmitter the signals sent/transmitted by/for each UE in each symbol period. {The work, however, did not focus on deriving a notion of convergence of the rKA-based schemes and relating it to amounts of performance and computational complexity.}

{
	In this paper, we reconsider the emulation of some canonical linear schemes by means of rKA, as originally proposed in \cite{Boroujerdi2018b}. We present a different perspective of analysis that aims to collect more insights about the algorithm capabilities and possible improvements, placing as an application scenario a site in which a BS composed by a large number of antennas but equipped with simple and inexpensive hardware is responsible for ensuring communication. Different from \cite{Boroujerdi2018b}, our aim is {\it to obtain the receive combining matrix through rKA and, then, acquire the transmit precoding matrix by applying the uplink-downlink (UL-DL) duality}. This leads to a presentation with a more general point of view with respect to UEs' mobility and prevents the BS from executing a rKA in each symbol period, as done in \cite{Boroujerdi2018b}. The analysis herein considers the emulation of the classical zero-forcing (ZF) and regularized zero-forcing (RZF) solutions for data detection, by seeing them as SLE problems that can be solved via rKA. The choice for these schemes is motivated by their suitable performance--complexity trade-off in classical M-MIMO scenario \cite{Bjornson2017} -- in which the computational complexity is not considered as a constraint at first. In particular, our contributions in this work are threefold:
	\begin{itemize}
		\item We show that large-term fading effects (pathloss and shadowing) are detrimental to the rKA functionality when this approach is applied to obtain the receive combining matrix, as revealed in \cite{CroisfeltRodrigues} by which we broadly introduced these findings in terms of performance only. This motivated us to propose modifications upon the initialization (forcing and hybridizing it) of a particular, called hereafter as parallel rKA-based scheme, improving its convergence, its average performance, and effectively decreasing its computational cost.
		\item Once the authors in \cite{Boroujerdi2018b} only presented the complexity related to rKA-based schemes in terms of big-$\mathcal{O}$ notation, we enhance the \textit{computational complexity analysis} considering complex scalar multiplications/divisions related to the multiplications/divisions operations of matrices and vectors (see the framework reported in \cite[App. B, p. 558]{Bjornson2017}). This enabled us to define upper bounds for the number of rKA iterations considering as basis the computational complexity difference between the proposed and available canonical schemes. We then bridge the gap between performance and computational complexity by characterizing numerically the number of rKA iterations required to our proposed scheme achieves a fraction of the performance acquired by classical schemes, in particular for the RZF filter. Using this, we compare the computational complexities of both approaches in a more particular extent, and identify significant behaviors upon the way that the algorithm converges when considering a case study based on a \textit{non-line of sight (NLoS) dense urban channel scenario}.
		\item More realistic M-MIMO channel conditions are included in our analysis intending to enlarge the insights provided on the convergence operation of rKA-based schemes for a NLoS dense urban scenario. Specifically, we develop a convergence analysis for the focused algorithm based on numerical results adopting the classical least-squares (LS) and minimum mean-squared error (MMSE) channel estimators, and spatially uncorrelated and correlated channels. Channel estimates are considered to demonstrate that imperfect CSI has interesting effects over the rate of convergence of rKA-based schemes that are subtly counter-intuitive. Spatial correlation is considered to embrace more practical conditions of system implementation.
	\end{itemize}	
	Conventionally, the LS channel estimator is often referred to as the most simple computational approach, on the other hand, sacrificing immensely performance, whereas MMSE is considered to be the best performance case, but in turn requires more processing. Spatial correlation is generated using a combination of the exponential correlation model \cite{Loyka2001} with long-term fading variations over the antenna array \cite{Bjornson2018a,Gao2015,Gao2015a}, thus, demonstrating the impact of two different sources of spatial correlation over the convergence of rKA-based schemes: (a) one arising from the proximity of antennas and (b) the other from the unequal contribution of power from the scatterers chaotically distributed along the propagation environment. 
}

This paper is then organized as follows. Section \ref{sec:systemmodel} describes the M-MIMO system model, defining the SE bound and the considered combining/precoding signal estimation solutions as optimization problems, and {the key properties of the KA and its random variant, the rKA.} The focused approach (rKA) is applied as a way to solve the combining/precoding optimization problems in Section \ref{sec:KA:precoding/decoding}. Section \ref{sec:complexity} evaluates the computational complexity of the proposed rKA-based scheme. Numerical results aiming at a better characterization of the rKA-based scheme in a much more realistic M-MIMO context are demonstrated in Section \ref{sec:results}. The main conclusions are summarized in Section \ref{sec:conc}. 

\textit{Notations}: Italic letters denote scalars, whereas boldface uppercase and lowercase represent matrices and column vectors, respectively. The superscript $(\cdot)^{\scriptscriptstyle \mathrm{H}}$ denotes the Hermitian transpose and $\mathcal{N}_{\mathbb{C}}(\cdot,\cdot)$ stands for the circularly symmetric complex Gaussian distribution. $\mathbb{E}\{x\}$ holds for the expected value of a random variable $x$. The $N \times N$ identity matrix is indicated by $\mathbf{I}_{N}$, whereas a vector of $N$ and a matrix of $N\times{N}$ zero elements are represented, respectively, as $\mathbf{0}_N$ and $\mathbf{0}_{N\times{N}}$. $\lVert\cdot\rVert_{2}$ and $\lVert\cdot\rVert_{\mathrm{F}}$ are the $l_{2}$-norm and Frobenius norm, respectively. The inner product between two vectors is denoted as: $\langle\cdot,\cdot\rangle$. The concatenation is represented by $[\cdot,\cdot]$ and the element in row $i$ and column $j$ of $\mathbf{A}$ is {selected} as $[\mathbf{A}]_{i,j}$. In particular, the $i$th element of a vector $\mathbf{a}$ can be seen as $a_{i}$. {Operator $\mathrm{tr}(\cdot)$ stands for the trace of a matrix.}

\section{{Massive MIMO System Model and Kaczmarz-Based Algorithms}}\label{sec:systemmodel}
In this section, a canonical single-cell M-MIMO is presented, aiming to introduce the SE bound utilized to evaluate the performance of signal processing techniques and to state {the traditionally implemented combining/precoding} steps as optimization problems{, which can be viewed as SLEs solved later through the use of the Kaczmarz premises. Basic concepts behind the classical KA and its randomized variant are also given at the end of this section.} The setup considered here is comprised of a BS equipped with $M$ antennas serving $K$ \textit{single-antenna} UEs. {We then make use of block-fading model \cite{Marzetta2016,Bjornson2017}, in which the wireless channels linking UEs and BS antennas are constant during a coherence block of $\tau_c$ symbols, but varying independently from subsequent blocks.} In addition, it is adopted a correlated Rayleigh fading model whereby the channel can be written as $\mathbf{g}_{k}\sim\mathcal{N}_{\mathbb{C}}({\mathbf{0}_{M}},{\mathbf{R}_{k}})$. The covariance matrix $\mathbf{R}_{k}\in\mathbb{C}^{M \times M}$ thus embodies effects such as pathloss and spatial channel correlation, where both effects corresponds to long-term propagation phenomena, while the complex Gaussian distribution stands for the short-term fading \cite{Bjornson2018a}. At this moment, one can assume a pilot training phase to obtain the CSI that follows the features of the TDD scheme \cite{Marzetta2016,Bjornson2017}. The pilot sequences attributed to the UEs are regarded as mutually orthogonal, since it is desired to combat the intra-cell interference. As a consequence of this process, the BS acquires the estimated channel matrix $\hat{\mathbf{G}}\in\mathbb{C}^{M\times K}$ being $\hat{\mathbf{g}}_{k}\in\mathbb{C}^{M}$ its $k$th column, which represents the estimated channel vector of UE $k$. {Note that the CSI perfectly known at the receiver can never be attained in practical wireless systems, i.e., $\hat{\mathbf{G}} = {\mathbf{G}}$, therefore, a more practical assumption is to consider a channel estimation approach being running on the BS side, as the LS or MMSE estimator. We refer to \cite{Marzetta2016,Bjornson2017} for a more detailed description of such procedures.}

\subsection{Uplink Data Transmission}
In the UL, the BS receives the data symbols transmitted by all UEs, yielding the following linearly combined signal: $\mathbf{y} = \sqrt{\rho^{\mathrm{ul}}} \sum_{i = 1}^{K} \mathbf{g}_{i}{s}_{i}+\mathbf{n}$; where $\rho^{\mathrm{ul}}$ is the normalized UL transmit power, ${s}_{i} \sim \mathcal{N}_{\mathbb{C}}({0},{1})$ stands for the data symbol sent by UE $i$, and $\mathbf{n}\in\mathbb{C}^{M} \sim \mathcal{N}_{\mathbb{C}}({\mathbf{0}_{M}},{\mathbf{I}_{M}})$ denotes the receiver noise. The performance of a selected receive combining vector, {$\mathbf{v}_{k}\in\mathbb{C}^{M}$}, will be evaluated using the UL SE metric obtained through the application of the so-called use-and-forget technique\footnote{This denomination stems from the exploitation of the channel estimates only to compute the receive combining vector, whereas the CSI is not used in the detection process.}. The net ergodic and lower bounded UL SE for UE $k$ is then \cite{Marzetta2016,Bjornson2017}
\begin{equation}
\underline{\mathrm{SE}}^{\mathrm{ul}}_{k} = \dfrac{\tau_{\mathrm{ul}}}{\tau_{\mathrm{c}}}\log_{2}\left(1 + \underline{\gamma}^{\mathrm{ul}}_{k} \right) \ \text{[bit/s/Hz]},
\label{eq:sekul}
\end{equation}
with the effective signal-to-interference-plus-noise ratio (SINR) given as
\begin{equation}
\underline{\gamma}^{\mathrm{ul}}_{k} = \dfrac{\rho^{\mathrm{ul}} \lvert{\mathbb{E}\{\mathbf{v}^{\scriptscriptstyle \mathrm{H}}_{k}\mathbf{g}_{k}}\}\rvert^2}{\rho^{\mathrm{ul}} 		\textstyle\sum_{i=1}^{K}\mathbb{E}\{\lvert\mathbf{v}^{\scriptscriptstyle \mathrm{H}}_{k}\mathbf{g}_{i}\rvert\}^2 - \rho^{\mathrm{ul}} \lvert{\mathbb{E}\{\mathbf{v}^{\scriptscriptstyle \mathrm{H}}_{k}\mathbf{g}_{k}}\}\rvert^2 + \mathbb{E}\{\lVert\mathbf{v}_{k}\rVert^2_{2}\}},
\label{eq:ulsinrk}
\end{equation}
where $\tau_{\mathrm{ul}}$ is the fraction of the coherence interval spent in UL. Supposing that only the UL {phase} is occurring, $\tau_{\mathrm{ul}}=\tau_{\mathrm{c}}-\tau_{\mathrm{p}}$ being $\tau_{\mathrm{p}}$ the length of the pilot training phase. The expectations in \eqref{eq:ulsinrk} are taken with respect to all randomness related to the channel estimates and combining schemes.

\subsection{Downlink Data Transmission}\label{sec:sub:down}
The BS sends data to its respective UEs during the DL, where each data bearing has to be precoded with the aim of spatially direct it towards the desired UE. This directional beam control is performed through a transmit precoding vector {given as $\mathbf{w}_{k}\in\mathbb{C}^{M}$ for the $k$th UE.} Thus, the first stage of DL is to steer the information signal to be transmitted by the BS via a precoding process, {here with equally distributed power among the UEs}, stated as $\textstyle\mathbf{x}=\sum_{i=1}^{K}\mathbf{w}_{i}\varsigma_{i}$; where $\varsigma_{i}\sim\mathcal{N}_{\mathbb{C}}({0},{1})$ is the message signal sent by the BS for UE $i${, and} $\lVert\mathbf{w}_{i}\rVert^{2}_{2} = 1$ such that the transmit precoding process does not interfere in the signal power. The precoded signal is then transmitted by the BS {where, recalling that once TDD mode is assumed,} the UL and DL channels are {comprehended as reciprocal} within each $\tau_{\mathrm{c}}$-block. Keeping this in mind, UE $k$ receives: ${y}_{k} = \sqrt{\rho^{\mathrm{dl}}}\mathbf{g}^{\scriptscriptstyle \mathrm{H}}_{k}\mathbf{x}+{n}$; where $\rho^{\mathrm{dl}}$ is the normalized DL transmit power. It can be seen that each UE is affected by all UEs' precoding vectors and, for this reason, the selection of $\mathbf{w}_{i}$ becomes challenging. A heuristic way to resolve the selection of the precoding vectors across the UEs is given through the use of the UL-DL duality\cite{Marzetta2010,Marzetta2016,Bjornson2017}, which will be better discussed and deeply exploited in the sequel. The UE $k$ can estimate its respective message signal by computing the average precoding channel $\mathbb{E}\{\mathbf{g}^{\scriptscriptstyle \mathrm{H}}_{k}\mathbf{x}\}$. This procedure allows the derivation of a hardening bound similar to \eqref{eq:sekul} for the DL SE, given in \cite[p. 317]{Bjornson2017}, with an alike definition of $\tau_{\mathrm{dl}}$ and {which depends greatly on achieving the hardening effect of the channel.} The SE bound of the DL is not directly evaluated here, since the UL evaluation is sufficient to demonstrate the effectiveness of the {proposed contributions} without any loss of generality.

\subsection{Combining/Precoding Schemes: Problem Statement}
Throughout this section, the ZF and RZF combining schemes are presented and then, relying on the UL-DL duality, their respective precoding vectors are obtained. For ease of exposition, the receive combining and transmit precoding matrices are {defined} as $\mathbf{V}\in\mathbb{C}^{M\times{K}}$ and $\mathbf{W}\in\mathbb{C}^{M\times{K}}$ with $\mathbf{v}_{k}$ and $\mathbf{w}_{k}$ being their $k$th columns, respectively. {It is furthermore demonstrated that} the introduced canonical schemes can be seen as the solution of optimization problems. As shall be {discussed}, the optimization problems can be treated as SLEs and thus can be naturally solved through rKA, as conducted in section \ref{sec:KA:precoding/decoding}, giving rise to the rKA-based signal processing schemes. {Note that we are using the word "scheme" in the plural because different ways of interpreting mathematically the rKA can be used to derive different schemes. Our focus, however, is on a particular way, namely the parallel rKA-based RZF scheme, as described in Section \ref{sec:sub:PARL}; the RZF term is used to identify that the underlying SLE being solved by the rKA is based on the classical RZF solution.}

The RZF is a sophisticated approach that takes into account the regularization of the estimated channel matrix by the {inverse of the} signal-to-noise ratio (SNR). This regularization factor can mathematically aid the inversion of $\hat{\mathbf{G}}$ and, physically, it weights the rates of interference suppression and desired signal maximization provided by this scheme. The RZF is given as \cite{Bjornson2017}
\begin{equation}
\mathbf{V}^{\scriptscriptstyle \mathrm{RZF}}= \hat{\mathbf{G}}\left(\hat{\mathbf{G}}^{\scriptscriptstyle \mathrm{H}}\hat{\mathbf{G}}+\xi\mathbf{I}_{K}\right)^{-1}
\label{eq:rzfcombiner}
\end{equation}
where $\xi=\frac{1}{\rho^{\mathrm{ul}}}$ is the regularization factor. This approach is recommended to be used in scenarios where a good channel estimation is achieved, and the inter-cell interference is not duly strong in a multi-cell case.

Notice that if the SNR is high, $\xi\rightarrow0$, and/or a large $M$ is applied, $\textstyle \mathrm{tr}(\hat{\mathbf{G}}^{\scriptscriptstyle \mathrm{H}}\hat{\mathbf{G}})\gg\mathrm{tr}(\xi\mathbf{I}_{K})$, RZF can be approximated as \cite{Marzetta2016,Bjornson2017}
\begin{equation}
\mathbf{V}^{\scriptscriptstyle \mathrm{ZF}}=\hat{\mathbf{G}}\left(\hat{\mathbf{G}}^{\scriptscriptstyle \mathrm{H}}\hat{\mathbf{G}}\right)^{-1}.
\label{eq:zfcombiner}
\end{equation}
which is called as the ZF scheme. The term ZF comes from the fact that $\hat{\mathbf{G}}^{\scriptscriptstyle \mathrm{H}}\mathbf{V}^{\scriptscriptstyle \mathrm{ZF}}=\mathbf{I}_{K}$, {noting} that $\mathbf{V}^{\scriptscriptstyle \mathrm{ZF}}$ is the pseudo-inverse of $\hat{\mathbf{G}}^{\scriptscriptstyle \mathrm{H}}$. Consequently, the ZF aims to eliminate the intra-cell interference while maintains the power level of the desired signals, but it can sometimes increase the noise power when the channel is ill-conditioned. When the receive combining is applied, namely, $(\mathbf{V}^{\scriptscriptstyle \mathrm{ZF}})^{\scriptscriptstyle \mathrm{H}}\mathbf{y}$, however, the real channels are different from the estimated ones, {resulting in} residual interference and, consequently, degraded performance \cite{Bjornson2017}. \textit{It is worth mentioning that the ZF can be comprehended as the RZF case with $\xi=0$}.

In order to separate and estimate the signals of each UE, the receive combining scheme is applied over the received signal in the BS. This can be expressed as
\begin{equation}
\hat{\mathbf{s}}=\mathbf{V}^{\scriptscriptstyle \mathrm{H}}\mathbf{y},
\label{eq:signalest}
\end{equation}
where $\hat{\mathbf{s}}\in\mathbb{C}^{K}$ is the vector with the estimates of the data symbols transmitted by all UEs in one communication instant, and $\mathbf{V}$ can be either the RZF or the ZF filter. The signal estimation problem handled above can be stated as the following optimization problem for the RZF scheme:
\begin{equation}
\arg\min_{\boldsymbol{\varrho}\in\mathbb{C}^{K}} \lVert\hat{\mathbf{G}}\boldsymbol{\varrho}-\mathbf{y}\rVert^{2}_{2} + \xi \lVert{\boldsymbol{\varrho}}\rVert^{2}_{2}.
\label{eq:optz-problem}
\end{equation}
In the sense of minimizing the mean-squared error (MSE), it is well-known that the optimal solution for this problem is the RZF scheme when the regularization factor is the inverse of the SNR (see \cite{Bjornson2017,Boroujerdi2018b} for proof). The same procedure can be done for the ZF when considering $\xi=0$, however, \textit{an analysis focused on the most general case characterized here by the RZF scheme will be assumed hereafter}.

The optimization problem given in \eqref{eq:optz-problem} can be rewritten compactly as $\lVert\mathbf{B}\boldsymbol{\varrho}-\mathbf{y}_{0}\rVert^{2}_{2}$, where $\textstyle \mathbf{B}=[\hat{\mathbf{G}};\sqrt{\xi}\mathbf{I}_{K}]$ is an $(M+K)\times K$ matrix with the estimated channel matrix and the factorized regularization factor, and $\textstyle \mathbf{y}_{0}=[\mathbf{y};\mathbf{0}_{K}]$ is an $(M+K)$-dim vector with the received signal and a zero vector. $\hat{\mathbf{s}}$ can thus be expressed as
\begin{equation}
\hat{\mathbf{s}}=(\mathbf{B}^{\scriptscriptstyle \mathrm{H}}\mathbf{B})^{-1}\mathbf{B}^{\scriptscriptstyle \mathrm{H}}\mathbf{y}_{0}=\left(\hat{\mathbf{G}}^{\scriptscriptstyle \mathrm{H}}\hat{\mathbf{G}}+\xi\mathbf{I}_{K}\right)^{-1}\hat{\mathbf{G}}^{\scriptscriptstyle \mathrm{H}}\mathbf{y}
\label{eq:compact-version-of-s}.
\end{equation}
{These} definitions will be used to apply the rKA as an alternative way to solve the optimization problem seeing it as an SLE.

As anticipated before, the selection of precoding vectors is a very complex problem to be solved, due to its {inter-dependence}, i.e., a precoding vector is affected by all the precoding vectors of the system. A heuristic way to solve it is the UL-DL duality{, a direct consequence of considering TDD mode,} by which the transmit precoding vector of UE $k$ can be computed as \cite{Bjornson2017}
\begin{equation}
\mathbf{w}_{k} = \dfrac{\mathbf{v}_{k}}{\lVert\mathbf{v}_{k}\rVert_{2}}
\label{eq:ul-dl-duality}
\end{equation}
for any receive combining scheme. Recall that the normalization is related to the desired fact that the transmit precoding vector should not interfere with the transmitted power. The above definition allows the BS to obtain the transmit precoding vectors for ZF, RZF, and also KA- or rKA-based schemes without the exact solution of the precoding problem. This result is deeply explored in the next sections.

\subsection{Kaczmarz and Randomized Kaczmarz Algorithms}\label{sec:rKA}
The mathematical concepts behind the KA are now briefly described based on the Kaczmarz's seminal work \cite{Kaczmarz1937}. We also introduce those underlying the rKA which is an extension of KA provided by Strohmer and Vershynin in \cite{Strohmer2006} that {enabled the derivation of an expected notion on the rate of convergence. This means that more reliable results can be get with rKA when compared to those obtained with the classical KA}. We consider, throughout this section, the canonical and generic SLE $\mathbf{A}\mathbf{x}=\mathbf{b}$, where $\mathbf{A}\in\mathbb{C}^{m\times{n}}$ is the matrix of constant coefficients, $\mathbf{x}\in\mathbb{C}^{n}$ is the vector of unknowns, and $\mathbf{b}\in\mathbb{C}^{m}$ is the vector of known offset coefficients. Note that, under the consideration of an underdetermined (UD) SLE ($m<n$), $\mathbf{A}$ is a full row-rank matrix, whereas for an OD SLE ($m>n$), $\mathbf{A}$ is a full column-rank matrix. Besides that, the SLE is deemed consistent, which means that it possesses at least one possible solution.

The KA can be written as \cite{Kaczmarz1937,Strohmer2006}
\begin{equation}
\mathbf{x}^{t+1}=\mathbf{x}^{t}+\dfrac{b_{i}-\langle\mathbf{a}_{i},\mathbf{x}^{t}\rangle}{\lVert{\mathbf{a}_i}\rVert^{2}_{2}}\mathbf{a}_{i},
\label{eq:KA:KA}
\end{equation}
where $t$ is the iteration index and $\mathbf{a}_{i}=([\mathbf{A}]_{i,:})^{\scriptscriptstyle \mathrm{H}}\in\mathbb{C}^{n}$ is the $i$th round-robin selected row of $\mathbf{A}$; similarly, $b_{i}$ is the $i$th element of $\mathbf{b}$. It is important to stay clear that the selected row $i$ is a deterministic variable for the classical KA and can then be computed circularly, for instance, as $i=\mathrm{mod}(t,m)+1$. We describe the KA in detail as follows. First, {the KA} is initialized with an arbitrary solution, $\mathbf{x}^{0}$, seeking to find the genuine answer, $\mathbf{x}$. It then takes a row $i$ in a specific KA iteration $t$ {that} is successively chosen from the order of the equations in $\mathbf{A}$, i.e., sweeping from $\mathbf{a}_{1}$ to $\mathbf{a}_{m}$. After that, KA performs the computation $b_{i}-\langle\mathbf{a}_{i},\mathbf{x}^{t}\rangle$, where notice that if $\langle\mathbf{a}_{i},\mathbf{x}^{t}\rangle=b_{i}$, the solution is not updated, resulting in $\mathbf{x}^{t+1}=\mathbf{x}^{t}$; if not, the KA projects the current approximate solution, $\mathbf{x}^{t}$, onto the hyperplane defined by $\langle\mathbf{a}_{i},\mathbf{x}^{t}\rangle = b_{i}$, thus securing the solution consistency and performing the update $\mathbf{x}^{t+1}$. These successive orthogonal projections tend to lead $\mathbf{x}^{0}$ to $\mathbf{x}$ for a suitable number of KA iterations. In summary, the KA is an efficient manner to find the solution of an SLE, {the method strives} to guide a random guess in the direction of the desired answer while obeying the maximum energy perturbation principle \cite{Boroujerdi2018b} {in relation to $\mathbf{A}$.}

Even though KA is an extremely powerful solver, mainly for OD SLEs, its rate of convergence is not duly known and is hence not satisfactorily explored. This occurs due to the KA's rate of convergence being a difficult metric to be estimated, since there is a strong dependence on the way the equations are arranged in $\mathbf{A}$; whose fact affects the so-called KA's update schedule, namely, the manner the equations are selected as the number of iterations grows. Eventually, some works identified that the random choice of the rows of $\mathbf{A}$ can guarantee a rigorous rate of convergence for the KA. From this observation, \cite{Strohmer2006} proposed the rKA scheme that assures an expected exponential rate of convergence, which is of capital importance to ensure a reliable KA solution in terms of reproducible results.

The rKA randomly and independently chooses the rows of $\mathbf{A}$ in its update schedule, founded on a metric that measures the relevance of the rows. In this case, the random KA-variant model described in \cite{Strohmer2006} can be expressed as
\begin{equation}
\mathbf{x}^{t+1}=\mathbf{x}^{t}+\dfrac{b_{r(t)}-\langle\mathbf{a}_{r(t)},\mathbf{x}^{t}\rangle}{\lVert{\mathbf{a}_{r(t)}}\rVert^{2}_{2}} \mathbf{a}_{r(t)},
\label{eq:KA:rKA}
\end{equation}
where the random picked row in {rKA} iteration $t$ is denoted as $r(t)\in\{1,2,\dots,m\}$. In particular, each $r(t)$ has a sample probability given as ${P_{r(t)}}\in[0,1]$ that embraces a vector with the specified sampling distribution of all rows equals to $\mathbf{p} = {(P_{1},\dots,P_{m})}\in\mathbb{R}^{m}_{+}$; {wherein one should observe that $\sum_{r(t)=1}^{m}{P_{r(t)}}=1$}. {This sample probability is defined to minimize the expected MSE at a given rKA iteration $t$, whose value is given by \cite{Strohmer2006}}
\begin{equation}
P_{r(t)}:=\dfrac{\lVert\mathbf{a}_{r(t)}\rVert^{2}_{2}}{\lVert\mathbf{A}\rVert^{2}_{\mathrm{F}}}.
\label{eq:KA:rowpdf}
\end{equation}
In fact, the above ${P_{r(t)}}$ is suboptimal\footnote{Optimization of the probability may not be useful for M-MIMO, since the CSI changes every $\tau_{\mathrm{c}}$. Indeed, as it shall be demonstrated, $\mathbf{A}$ is directly linked to $\hat{\mathbf{G}}$. Altogether, it will be required to solve optimization problems for each $\tau_{\mathrm{c}}$ wherein these solutions{, given the posed BS restrictions,} require unwanted hardware load \cite{Boroujerdi2018b,Dai2014}.} and {can be interpreted as} an amount that measures the fraction of energy of a given row with respect to the total energy of the system. This probability distribution, although suboptimal, has demonstrated effectiveness and, hence, will be adopted in this work. 
{Note, therefore, that the more relevant a row of $\mathbf{A}$, the greater the chance of this being chosen during the update schedule.} The main convergence result of the rKA established in \cite{Boroujerdi2018b} and \cite{Strohmer2006} is summarized by the following corollary.
\begin{corollary}[Expected Rate of Convergence of rKA] \label{col:rka-rate-of-convegence}
	The expected rate of convergence of the rKA is
	\begin{equation}
	\mathbb{E}\{\lVert{\mathbf{x}^{t}-\mathbf{x}^{\star}}\rVert^{2}_{2}\}\leq(1-\kappa_{\mathcal{X}}(\mathbf{A}))^{t}\lVert{\mathbf{x}^{0}-\mathbf{x}^{\star}}\rVert^{2}_{2},
	\label{eq:KA:rateofconvergence}
	\end{equation}
	with $\mathbf{x}^{\star}$ being the optimal solution and $\mathbf{x}^{0}$ an arbitrary initial guess. Moreover, $\kappa_{\mathcal{X}}(\mathbf{A})$ is nominated as the average gain and defined as
	\begin{equation}
	\kappa_{\mathcal{X}}(\mathbf{A}):=\min_{\boldsymbol{\vartheta}\in\mathcal{X},\boldsymbol{\vartheta}\neq0}\dfrac{\lVert\mathbf{A}\boldsymbol{\vartheta}\rVert^{2}_{2}}{\lVert\mathbf{A}\rVert^{2}_{\mathrm{F}}\lVert\boldsymbol{\vartheta}\rVert^{2}_{2}},
	\label{eq:KA:kappa}
	\end{equation}
	{where $\mathcal{X}\subset\mathbb{C}^{M+K}$ is the subspace generated by the columns of $\mathbf{A}$. The important observations here is} that the above metric is totally independent of the sample probability and is proportional to the condition number\footnote{Remember that $\lVert\mathbf{A}$$\rVert^{2}_{2}\lVert\mathbf{A}^{-1}\rVert^{2}_{\mathrm{F}}$ is usually defined as the condition number of $\mathbf{A}$.} of $\mathbf{A}$.
\end{corollary}
\begin{proof}
	The proof can be found in \cite[p. 11]{Boroujerdi2018b}.
\end{proof}

To gain further insights, the above corollary {defines} in \eqref{eq:KA:rateofconvergence} the quantity of rKA iterations (projections) required to take the initial solution $\mathbf{x}^{0}$ closer to $\mathbf{x}^{\star}$, based on an average gain provided by the rKA. This convergence is exponentially reaching the average inaccuracy limit that is given by the left-hand side of \eqref{eq:KA:rateofconvergence}. {Note that the higher the $(1-\kappa_{\mathcal{X}}(\mathbf{A}))$}, more rKA iterations are needed to lower the initial error $\lVert\mathbf{x}^{0}-\mathbf{x}^{\star}\rVert^{2}_{2}$ towards the desired boundary. It is thus desirable to have a high value of $\kappa_{\mathcal{X}}(\mathbf{A})$ so as to shrink the number of rKA iterations necessary to achieve, suitably, the expected error limit. This means that there is an irrefutable connection between the $\kappa_{\mathcal{X}}(\mathbf{A})$ and the convergence of the rKA. As a result, the best-case of convergence is represented by the maximum value that the average gain can assume, while its minimum stands for the worst-case scenario of convergence. It is also worth mentioning explicitly that the number of rKA iterations is inversely proportional to the average gain. Because of the above convergence result, which brings additional reliability, our focus, from now on, is the use of rKA in the design of signal processing schemes for M-MIMO based on the Kaczmarz methodology\footnote{One must stay clear that procedures to obtaining the combining and precoding matrices via the rKA can also be implemented with the KA, however, the converge will immensely rely on the way that the equations are ordered in the SLE to be solved.}.

\section{Obtaining the Combining/Precoding Matrices via Randomized Kaczmarz Algorithm}\label{sec:KA:precoding/decoding}
The rKA is now presented as a valid approach to solve the {optimization problem of the UL signal estimation established in \eqref{eq:optz-problem}}, which can be interpreted as if the rKA was emulating the solution offered by the RZF scheme. The precoding matrix can then be derived using the UL-DL duality, described in \eqref{eq:ul-dl-duality}, without loss of generality. An alternative way to solve the optimization problem considered in \eqref{eq:optz-problem} is its statement through the resolution of an OD SLE given as $\mathbf{B}\boldsymbol{\varrho}=\mathbf{y}_{0}$. Notice that this SLE is OD, since the number of equations, $M+K$, is greater than the number of unknowns, $K$, in a typical M-MIMO scenario {\cite{Marzetta2010}}. Moreover, observe that this problem is not consistent because of the arbitrariness inserted by the receiver noise in the signal received by the BS, $\mathbf{y}$. This makes that the OD SLE possesses an infinite number of solutions. If the rKA is applied to solve this problem, a convergence bias or residual error is eventually introduced in the solution due to the inconsistency. In \cite{Boroujerdi2018b}, this problem was circumvented by splitting $\mathbf{B}\boldsymbol{\varrho}=\mathbf{y}_{0}$ with $\boldsymbol{\varrho}^{\star}=\hat{\mathbf{s}}$ into two stages. The first strives to remodel the OD SLE to a consistent form, leading to an UD SLE; then, another OD SLE is derived based on the genuine stated problem ($\mathbf{B}\boldsymbol{\varrho}=\mathbf{y}_{0}$). We present this in depth below.

\subsection{The Two Steps in Solving the SLE}
The first goal is to obtain an estimate of $\mathbf{y}_{0}$, {the term incorporating} the inconsistency. To this end, $\hat{\mathbf{y}}_{0} \in \mathbb{C}^{(M+K)}$ is {first} defined as
\begin{equation}
\hat{\mathbf{y}}_{0}:=\mathbf{B}\hat{\mathbf{s}}\overset{(a)}{=}\mathbf{B}(\mathbf{B}^{\scriptscriptstyle \mathrm{H}}\mathbf{B})^{-1}\hat{\mathbf{G}}^{\scriptscriptstyle \mathrm{H}}\mathbf{y},
\label{eq:esty0}
\end{equation}
where in ($a$) the relation exposed in \eqref{eq:compact-version-of-s} was applied. By making a parallel with the original SLE, $\mathbf{B}^{\scriptscriptstyle \mathrm{H}}\mathbf{y}_{0}=\boldsymbol{\varrho}$, the $\hat{\mathbf{y}}_{0}$ can be estimated via the following SLE:
\begin{equation}
\mathbf{B}^{\scriptscriptstyle \mathrm{H}}\hat{\mathbf{y}}_{0}=(\mathbf{B}^{\scriptscriptstyle \mathrm{H}}\mathbf{B})(\mathbf{B}^{\scriptscriptstyle \mathrm{H}}\mathbf{B})^{-1}\hat{\mathbf{G}}^{\scriptscriptstyle \mathrm{H}}\mathbf{y} = \hat{\mathbf{G}}^{\scriptscriptstyle \mathrm{H}}\mathbf{y}.
\label{eq:ud-sle}
\end{equation}
The above SLE is UD inasmuch as the number of equations, $K$, is smaller than the unknowns amount, $M+K$. Even more, it is consistent {since} all the solutions are within the subspace generated by the columns of $\mathbf{B}^{\scriptscriptstyle \mathrm{H}}$. Regardless of not being optimum, $\hat{\mathbf{y}}_{0}$ can be estimated through rKA. That being the case, it is natural to derive a second SLE given as $\mathbf{B}\hat{\mathbf{s}}=\hat{\mathbf{y}}_{0}$, wherein both OD and consistency conditions are comprised. The solution of this second {SLE based on rKA}, however, is not strictly necessary, seeing that the estimate of $\hat{\mathbf{s}}$ can be acquired through the $K$ last rows of $\hat{\mathbf{y}}_{0}$. This means that solely part of the UD SLE has to be solved, which can be easily seen when $\hat{\mathbf{y}}_{0}$ is replaced in $\mathbf{B}\hat{\mathbf{s}}=\hat{\mathbf{y}}_{0}$ as its definition in \eqref{eq:esty0}; doing that we get
\begin{equation}
\begin{split}
\mathbf{B}\hat{\mathbf{s}}&=\hat{\mathbf{y}}_{0}=\mathbf{B}(\mathbf{B}^{\scriptscriptstyle \mathrm{H}}\mathbf{B})^{-1}\hat{\mathbf{G}}^{\scriptscriptstyle \mathrm{H}}\mathbf{y} \\
\hat{\mathbf{s}}&=\mathbf{B}^{\scriptscriptstyle \mathrm{H}}\mathbf{B}(\mathbf{B}^{\scriptscriptstyle \mathrm{H}}\mathbf{B})^{-1}\hat{\mathbf{G}}^{\scriptscriptstyle \mathrm{H}}\mathbf{y}\overset{(a)}{=}\mathbf{B}^{\scriptscriptstyle \mathrm{H}}\hat{\mathbf{y}}_{0},
\end{split}
\label{eq:inception}
\end{equation}
where in ($a$) it was used the equality in \eqref{eq:ud-sle}.

\subsection{Parallel rKA-Based RZF Scheme}\label{sec:sub:PARL}
The receive combining matrix can be obtained through the implementation of $K$ rKAs in parallel, in which the $k$th one is solving an SLE of the form: $\mathbf{B}^{\scriptscriptstyle \mathrm{H}}\mathbf{c}=\mathbf{e}_{k}$. By already adopting the Kaczmarz notation, $\mathbf{c}^{t}=[\mathbf{u}^{t},\mathbf{z}^{t}]\in\mathbb{C}^{(M+K)}$ with $\mathbf{u}^{t}\in\mathbb{C}^{M}$ and $\mathbf{z}^{t}\in\mathbb{C}^{K}$. Further, $\mathbf{e}_{k}\in\mathbb{C}^{K}$ is the canonical basis vector with $[\mathbf{e}_{k}]_{k} = 1$ and $0$ otherwise. This analogous problem is essentially founded in the UD SLE exposed in \eqref{eq:ud-sle} and the property given in \eqref{eq:inception}, being its parallel resolution properly described in Algorithm \ref{algthm:parlrka}. It is remarkable to see that, in the parallelized form, $\mathbf{e}_{k}$ is replacing $\hat{\mathbf{G}}^{\scriptscriptstyle \mathrm{H}}\mathbf{y}$ from its underlying SLE in \eqref{eq:ud-sle}, whose replacement activates only the resolution for the equations related to the $k$th rKA. Elaborating further, this aspect suggests {that each iteration of the outer loop (\textit{step 4} to \textit{step 18}) in Algorithm \ref{algthm:parlrka} with respect to $k$ is solving an SLE} to obtain the receive combining vector associated with the UE $k$. {Observe that} the $k$th rKA solution does not actually yields $\mathbf{v}_{k}$, but it provides the vector $\mathbf{d}_{k}\in\mathbb{C}^{K}$ that is given solely by the last $K$ rows of $\mathbf{c}^{t}$, i.e., $\mathbf{z}^{t}$; thereby, $\mathbf{v}_{k}=\hat{\mathbf{G}}\mathbf{d}_{k}$ (or $\mathbf{v}^{\scriptscriptstyle \mathrm{rKA}}_{k}$ in order to emphasize that the combining vector is estimated via rKA).

Notably, the execution of each rKA {, i.e., each $k$th iteration of the outer loop defined from \textit{step 4} to \textit{step 18}} can be envisioned to run in different hardware units {working concurrently}, where each unit is independently or quasi-independently reaching its own receive combining vector. This independence is very related to the way that the randomness of the update schedule in \eqref{eq:KA:rKA} is interpreted at a given context, i.e., how the random selection of rows are associated across the hardware units {running in parallel}. With this in mind, two possibilities are plausible to give satisfactory results: (a) all the $K$ rKAs can share the random rows, $r(t)$s, and (b) each rKA has its self-realized arbitrariness. This discussion is better deepened in Section \ref{sec:complexity}, which treats about computational complexity; however, realize {at this moment} that both approaches provide suitable convergence{, such as that} established in Corollary 1.

\begin{algorithm}[htp]
	\centering
	\caption{Parallel (PARL) approach of rKA to estimate the RZF receive combining matrix}
	\label{algthm:parlrka}
	\begin{algorithmic}[1]
		\State \textbf{Input:} Estimated channel matrix $\hat{\mathbf{G}}\in\mathbb{C}^{M\times{K}}$, inverse level of SNR $\xi$, number of UEs $K$, and number of rKA iterations $T_{\scriptscriptstyle\mathrm{rKA}}$.
		\State \textbf{Initialization:} Specify $\mathbf{D}^{\scriptscriptstyle \mathrm{rKA}}\in\mathbb{C}^{K\times{K}}= \mathbf{0}_{K\times{K}}$ as the factorized version of $\hat{\mathbf{G}}^{\scriptscriptstyle \mathrm{H}}\mathbf{V}^{\scriptscriptstyle \mathrm{rKA}}$.
		\State \textbf{Procedure:}
		\For{$k \leftarrow 1$ {\bf to} $K$} 
		\State Initialize the state vectors $\mathbf{u}^{t}\in\mathbb{C}^{M}$ and $\mathbf{z}^{t}\in\mathbb{C}^{K}$ with $\mathbf{u}^{0}=\mathbf{0}_{M}$ and $\mathbf{z}^{0}=\mathbf{0}_{K}$.
		\State Compute the canonical basis, $\mathbf{e}_{k}\in\mathbb{C}^{K}$, where $[\mathbf{e}_{k}]_{k}=1$ and $[\mathbf{e}_{k}]_{j}=0$, $\forall{j}\neq{k}$.
		\For{$t \leftarrow 0$ {\bf to} {$T_{\scriptscriptstyle\mathrm{rKA}}-1$}}
		\If {$t = 0$}
		\State Pick the $k$th row of $\hat{\mathbf{G}}^{\scriptscriptstyle \mathrm{H}}$, $\hat{\mathbf{g}}^{\scriptscriptstyle \mathrm{H}}_{r(t)}\in \mathbb{C}^{1\times M}$ for $r(t)=k$.
		\Else
		\State Pick the $r(t)$th row of $\hat{\mathbf{G}}^{\scriptscriptstyle \mathrm{H}}$, $\hat{\mathbf{g}}^{\scriptscriptstyle \mathrm{H}}_{r(t)}\in \mathbb{C}^{1\times M}$, with $r(t) \in \{1,2,\dots,K\}$ drawn based on $\mathbf{p}$ where $P_{r(t)}= \frac{\lVert\hat{\mathbf{g}}_{r(t)}\rVert^{2}_{2} + \xi}{\lVert\hat{\mathbf{G}}\rVert^{2}_{\mathrm{F}} + K\xi}$. 
		\EndIf
		\State Compute the residual $\eta^{t}:=\frac{[\mathbf{e}_{k}]_{r(t)}-\langle\hat{\mathbf{g}}_{r(t)},\mathbf{u}^{t}\rangle - \xi z^{t}_{r(t)}}{\lVert\hat{\mathbf{g}}_{r(t)}\rVert^{2}_{2} + \xi}$.
		\State Update $\mathbf{u}^{t+1} = \mathbf{u}^{t} + \eta^{t}\hat{\mathbf{g}}_{r(t)}$.
		\State Update $z^{t+1}_{r(t)} = z^{t}_{r(t)} + \eta^{t}$ and repeat the other positions: $z^{t+1}_{j} = z^{t}_{j}, \ \forall j \neq r(t)$.
		\EndFor 
		\State Update $\left[\mathbf{D}^{\scriptscriptstyle\mathrm{rKA}}\right]_{:,k}$ = $\mathbf{z}^{T_{\scriptscriptstyle\mathrm{KA}}-1}$.
		\EndFor
		\State \textbf{Output:} $\mathbf{V}^{\scriptscriptstyle\mathrm{rKA}}=\hat{\mathbf{G}}\mathbf{D}^{\scriptscriptstyle\mathrm{rKA}}$.	
	\end{algorithmic}
\end{algorithm}

In spite of the fact that Algorithm \ref{algthm:parlrka} was assessed indirectly\footnote{{Recall that the authors of \cite{Boroujerdi2018b} discuss the rKA as an approach to obtain directly each element of $\hat{\mathbf{s}}$, motivated by slow mobility scenarios. Herein, we analyze a most generic case, where obtaining individually the combining/precoding matrices is of paramount importance, {being this knowledge used until the end of a coherence interval to estimate (indirectly) the elements of $\hat{\mathbf{s}}$.}}} in \cite{Boroujerdi2018b}, some changes are included here. To put them in context, it is reasonable to infer that the {sample probability $P_{r(t)}$} is generally impacted by the pathloss and shadowing due to {their} relation with the estimated channel powers (see the computation of $P_{r(t)}$ in \textit{step 11} of Algorithm \ref{algthm:parlrka}). Note that if $\mathbf{p}$ is detrimentally affected, the rKA's update schedule is also degraded. At this moment, one should note that if row $k$ is not selected in the $k$th {(standard)} rKA{, that is, if \textit{step 9} did not exist and $r(t)$ never equals $k$}, the solution update would not even occur, making the desired convergence unattainable. This also happens because of the use of zero vectors to initialize the state variables in $\mathbf{c}^{0}$. The zero vectors, however, are desirable, since other arbitrary values, like a vector of ones, can cause a power bias over the rKA-estimated receive combining vectors. Mathematically, this has also much to do with the fact that the initialization vector must be chosen in such a way that it belongs to the column space of the coefficient matrix {($\mathbf{B^{\scriptscriptstyle\mathrm{H}}}$ in this case)}; which property is not being satisfied by a vector of ones{, for example}. In summary, the UEs close to the BS (center-UEs) have significant higher values of $P_{r(t)}$ than those for UEs located at the edge of the cell (edge-UEs). This disparity leads to a poor or impossible convergence of the combining vectors {related to edge-UEs in Algorithm \ref{algthm:parlrka} when ignoring by the time the existence of \textit{step 9}}. Section \ref{sec:nume:subsec:rowprob} gives numerical insights on the consequences {of these} long-term {channel} effects.

Striving to overcome this negative effect, we observed that a suitable and non-random manner to initialize the $k$th {rKA} and thus ensure its update in accordance with the column space of the coefficient matrix is to force the selection of the $k$th row in the first run; as performed in \textit{step 9} of Algorithm \ref{algthm:parlrka}. This proceeding {can also provide} a better average initial guess, which translates into a better average gain {that improves the {expected} rate of convergence presented in Corollary 1, as shall be evidenced in the sequel. What makes the proposed procedure an attractive way to initialize Algorithm 1 is the fact that edge-UEs can now fairly obtain their receive combining vectors, even with the tough preference given to the center-UEs in the computation of $P_{r(t)}$. {This gives for} the parallel (PARL) rKA-based RZF scheme {more robustness under} practical regimes and, consequently, {better performance due to improved} convergence towards the aimed solutions.} One can additionally conjecture that {the proposed} initialization approach, which can be seen as a \textit{hybrid algorithm} between KA and rKA, is not strictly necessary to UEs that have high values of ${P_{r(t)}}$, for the simple reason that they are more probable to occur when drawing $r(t)$s. As a result, an optimization problem to find which one of the $K$ {rKAs} has to be force-initialized could be suggested. {Note, however,} that the gain provided by this possible optimization is merely one more iteration for each {rKA}, and its finding may produce {a considerable high and undesired number of additional computations in the BS}. This unattractive optimization procedure was therefore not considered in this work. 

The convergence performance of Algorithm \ref{algthm:parlrka} is defined in Theorem \ref{thrm:RoC-PARL}. {For convenience, we also restated} some comments about the average gain bounds based on \cite{Boroujerdi2018b} in Remarks \ref{rmk:parl1} and \ref{rmk:parl2}. Recall from Section \ref{sec:rKA} that the limits of $\kappa_{\mathcal{X}}(\mathbf{B}^{\scriptscriptstyle \mathrm{H}})$ are inversely proportional to the respective bounds of the number of rKA iterations. These remarks will hence be important for the computational evaluation of rKA-based RZF schemes.
\begin{theorem}[PARL rKA-Based RZF Scheme: Rate of Convergence] \label{thrm:RoC-PARL}
	Consider the $k$th UD SLE $\mathbf{B}^{\scriptscriptstyle \mathrm{H}}\mathbf{c}=\mathbf{e}_{k}$, where $\mathbf{B}^{\scriptscriptstyle \mathrm{H}}\in\mathbb{C}^{K\times(M+K)}$ is the matrix of constant coefficients, $\mathbf{c}\in\mathbb{C}^{(M+K)}$ is the vector of unknowns, and $\mathbf{e}_{k}\in\mathbb{C}^{K}$ is the canonical basis. By applying the {rKA} to solve this system, as done in Algorithm \ref{algthm:parlrka}, it converges towards the optimal solution $\mathbf{c}^{\star}$ bounded by the following average error:
	\begin{equation}
	\mathbb{E}\{\lVert{\mathbf{c}^{t}-\mathbf{c}^{\star}}\rVert^{2}_{2}\}\leq(1-\kappa_{\mathcal{X}}(\mathbf{B}^{\scriptscriptstyle \mathrm{H}}))^{t}\lVert{\mathbf{c}^{(1)}-\mathbf{c}^{\star}}\rVert^{2}_{2},
	\end{equation}
	\begin{equation}
	\begin{split}
	\text{	with} \qquad		\kappa_{\mathcal{X}}(\mathbf{B}^{\scriptscriptstyle \mathrm{H}})=\dfrac{\lambda_{\min}(\mathbf{B}^{\scriptscriptstyle \mathrm{H}}\mathbf{B})}{\lVert\mathbf{B}\rVert^{2}_{\mathrm{F}}}\overset{(a)}{=}\dfrac{\lambda_{\min}(\hat{\mathbf{G}}^{\scriptscriptstyle \mathrm{H}}\hat{\mathbf{G}})+\xi}{\lVert\hat{\mathbf{G}}\rVert^{2}_{\mathrm{F}}+K\xi},
	\end{split}
	\label{eq:kappa:parl}
	\end{equation}
	where in ($a$), $\mathbf{B}$ is replaced by its definition stated in \eqref{eq:compact-version-of-s}. Note that $\mathbf{c}^{(1)}$ can be seen as the initial solution of the {rKA}, which also provides an average gain based on the non-random KA:
	\begin{equation}
	\kappa^{(1)}_{\mathcal{X}}(\mathbf{B}^{\scriptscriptstyle \mathrm{H}})=\min_{\boldsymbol{\vartheta}\in\mathcal{X},\boldsymbol{\vartheta}\neq0}\dfrac{\lVert\mathbf{b}_{k}^{\scriptscriptstyle \mathrm{H}}\boldsymbol{\vartheta}\rVert^{2}_{2}}{\lVert\mathbf{b}_{k}\rVert^{2}_{2}\lVert\boldsymbol{\vartheta}\rVert^{2}_{2}},
	\end{equation}
	that is very dependent on the row index $k$. The gap $\mathbf{c}^{(1)}-\mathbf{c}^{\star}$ can {thus} be made smaller on average than the provided by the adoption of a zero or any other random solution as the initial one.
\end{theorem}

\begin{proof}
	The proof is given in Appendix B.
\end{proof}

\begin{remark}[Average Gain Bounds: Generic Correlated Covariance Matrices with Gaussian Estimated Channel Matrix] \label{rmk:parl1}
	If the elements of $\hat{\mathbf{G}}$ are distributed as $\mathcal{N}_{\mathbb{C}}({0},1)$\footnote{For the canonical channel estimators analyzed in the numerical results, this distribution does not hold, since the estimated powers are certainly smaller than the true ones. Notice as well that the bound is disregarding long-term fading effects.}, the eigenvalues of $\hat{\mathbf{G}}$ are upper limited by $1$ and then, if all eigenvalues are $1$, we have $\lVert\hat{\mathbf{G}}\rVert^{2}_{\mathrm{F}}=K$. Thereby, the average gain stated in \eqref{eq:kappa:parl} is limited as the following manner \cite{Boroujerdi2018b}:
	\begin{equation}
	\dfrac{\lambda_{\min}(\hat{\mathbf{G}}^{\scriptscriptstyle \mathrm{H}}\hat{\mathbf{G}})}{\lVert\hat{\mathbf{G}}\rVert^{2}_{\mathrm{F}}}\leq\dfrac{\lambda_{\min}(\hat{\mathbf{G}}^{\scriptscriptstyle \mathrm{H}}\hat{\mathbf{G}})+\xi}{\lVert\hat{\mathbf{G}}\rVert^{2}_{\mathrm{F}}+K\xi}\leq\dfrac{1}{K},
	\end{equation}
	where $\xi\geq{0}$. From the above inequality, it is possible to presume that the higher the $\xi$, the faster the convergence of Algorithm \ref{algthm:parlrka}, since $\kappa_{\mathcal{X}}(\mathbf{B^{\scriptscriptstyle \mathrm{H}}})$ also turns out to be higher. Consequently, the ZF filter case, i.e., $\xi=0$ is expected to support a slower convergence than the provided by emulating the RZF scheme, {both through rKA}\footnote{The results that will be obtained for the rKA emulating the RZF scheme can be easily expanded to the ZF case, but, since ZF has a worse rate of convergence than RZF, the expected computational gains must be reduced for this last scheme.}. Another important pondering is to consider the case that the elements of $\hat{\mathbf{G}}$ are distributed as $\mathcal{N}_{\mathbb{C}}({0},\psi_{k})$, {where $\psi_{k}\in[0,1]$ is the variance of the estimated channel linking the $k$th UE to the $m$th antenna of the BS {which} incorporates the estimation error. $\psi_{k}$ can thus be interpreted as a fraction or a percentage of the true channel energy, where, to give a sense of values, the normalized MSE (NMSE) of the LS channel estimator is near to $1$ for a SNR of $0$ dB, while of the MMSE is close to $0.2$ following results reported in \cite{Bjornson2017}.} In this case, it is possible to observe that the upper limit is tugged far away from the value of ${1}/{K}$, slowing down too much the convergence of the algorithm as discussed in Corollary \ref{col:rka-rate-of-convegence}. {This fact hinders} the application of the rKA in practical scenarios of interest, {where the BS does not know exactly the CSI of the UEs and needs to appeal to a method of channel estimation. Again, this is one of the reasons why we adopted the LS and MMSE channel estimates when searching a more suitable notion of convergence in Section \ref{sec:nume:subsec:roc:nbrofiterations}.} For the sake of clarity, observe that the diagonal elements of $\hat{\mathbf{G}}^{\scriptscriptstyle \mathrm{H}}\hat{\mathbf{G}}$ are estimates of the $k$th average estimated channel variance and the non-diagonal elements corresponds to the average estimate {of spatial correlation} between different estimated channels. The first varies too much with long-term and spatial correlation effects, while the latter rely mainly on spatial correlation similarities among UEs and estimation errors. 
\end{remark}

\begin{remark}[Average Gain Bounds: Uncorrelated Rayleigh Fading with Gaussian Estimated Channel Matrix] \label{rmk:parl2}
	When the uncorrelated Rayleigh fading channel modeling is considered{, i.e., $\mathbf{R}_{k}=\beta_{k}\mathbf{I}_{M}$, where $\beta_{k}$ is the average long-term fading coefficient (pathloss), disregarding the shadowing term, for UE $k$, and the elements of $\hat{\mathbf{G}}$ are distributed as $\mathcal{N}_{\mathbb{C}}({0},{1})$;} the limits of $\kappa_{\mathcal{X}}(\mathbf{B^{\scriptscriptstyle\mathrm{H}}})$ can be strictly established based on the random matrix theory. This is done by solving the characteristic polynomial of $\hat{\mathbf{G}}^{\scriptscriptstyle \mathrm{H}}\hat{\mathbf{G}}$, finding its minimum, and then applying the {statistical properties} of $\hat{\mathbf{G}}$; {making these steps, we get that \cite{Boroujerdi2018b}}:
	\begin{equation}
	\kappa_{\mathcal{X}}(\mathbf{B^{\scriptscriptstyle \mathrm{H}}})\geq\dfrac{\lambda_{\min}(\hat{\mathbf{G}}^{\scriptscriptstyle \mathrm{H}}\hat{\mathbf{G}})}{\lVert\hat{\mathbf{G}}\rVert^{2}_{\mathrm{F}}}=\dfrac{\left(\sqrt{M}-\sqrt{K}\right)^{2}}{MK} =\dfrac{\left(1-\sqrt{\frac{K}{M}}\right)^2}{K},
	\end{equation}
	where {${K}/{M}$} is the UE-BS antenna ratio, also known as the \textit{loading factor} in M-MIMO. From this result, the number of {rKA} iterations is {definitively} associated with the value of the loading factor. {Interestingly, as the value of ${K}/{M}$ is close to $1$, the average gain tends to be zero. This means that the number of rKA iterations needed to comprise the worst-case convergence condition goes to infinity, i.e., the convergence can be considered unattainable in that case}. One needs to keep in mind, however, that the UE-BS antenna ratio in M-MIMO is typically ${{K}/{M}}\in[0,0.5]$ \cite{Marzetta2016,Bjornson2017}.
\end{remark}

\section{{Computational Complexity Analysis}}\label{sec:complexity}
Since the canonical and proposed {schemes} are adequately described, we now turn our attention to find a fair way to analyze and compare their computational complexities. This analysis aims {to verify} if the rKA-based strategy {actually} paves the way {to computationally light signal processing schemes}. In this section, we will define functions that describe the computational costs of {the schemes based on} matrix-to-matrix multiplications/divisions, as well as the consideration of the same operations for vectors, in terms of complex scalar multiplications/divisions. This computational finding strategy is motivated by the framework presented in \cite[App. B, p. 558]{Bjornson2017}, wherein additions and subtractions are neglected with the premise that the latter are not as heavy as the former {to be handled from the point of view of hardware.} With this in mind, the examination of the computational cost {is divided} into the UL and DL phases, {as presented in the sequel. Before that, we first bring a discussion of two proposed possibilities of how the hardware on the BS side can be assembled depending on possible, practical needs that a network designer might face in applying the proposed scheme. To conclude this section, we derive upper bounds for the number of rKA iterations regarding the computational complexity gap between the canonical and proposed schemes.}

\subsection{{Discussing Possible Settings for BS Hardware}}

{
	The first most simple, natural setup is called as the \textit{flexible setting} (FLS) wherein the BS is comprised of a unique central processing unit that executes sequentially the PARL rKA-based RZF scheme defined in Algorithm \ref{algthm:parlrka}. The problem with this approach is basically the execution time of the algorithm until convergence of all receive combining vectors. Depending on the other activities of BS hardware and the channel quality, the number of rKA iterations $T_{\scriptscriptstyle\mathrm{rKA}}$ is likely to be swiftly limited, as well as the maximum gap between the true value and a rKA solution, as indicated in Corollary \ref{col:rka-rate-of-convegence}. This approach, although placing a performance limit, represents the most flexible solution, since an unknown number of UEs can be served by the BS. As shall be elucidated, this is the most outstanding characteristic of the FLS in relation to the other proposed setup alternative.
}

{
	The other setup is called as the \textit{time saving setting} (TSS) which is based on the observation that each $k$th iteration of the PARL rKA-based RZF scheme is in reality obtaining the factorized version of the receive combining vector associated with UE $k$, where $k = {1,2,\dots,K}$ (see \textit{step 4} to \textit{step 18} of Algorithm \ref{algthm:parlrka}). This allows us to assume that the BS hardware can be assembled with $K$ processing units connected in parallel, where each of them is designated to find the solution, through rKA, of the UD SLE problem in \eqref{eq:ud-sle} redefined for a single UE of interest. In other words, each processing unit therefore has a simple, cheap, and dedicated hardware responsible to obtain, via rKA, the results required to the BS get the receive combining vector related to one of the $K$ UEs being served by it. On one hand, this setup will reduce at most in a factor of $K$ the time required to the BS obtains the knowledge of the receive combining matrix. On the other hand, the cost and coordination complexity of the $K$ processing units can be high and unattractive in comparison with the ones of the FLS. Moreover, the number of UEs that the BS can serve will be limited by the number of processing units installed. This setting is therefore advised to be used when the number of UEs in the BS coverage area is commonly small, and the channel effects place very hard constraints to the value of the coherence time $\tau_{\mathrm{c}}$. 
}

{
	Before proceeding, we want to emphasize that the computational cost functions acquired below are well founded for both FLS and TSS, and the selection of which setting to use depends mainly on the conditions established by the location where the communication services need to be installed.
}
\subsection{{Uplink Computational Cost}}
{
	From Algorithm \ref{algthm:parlrka}, it can be observed that the heavier computations are related to acquire the sampling distribution $\mathbf{p}$ in \textit{step 11}, and the residual $\eta^{t}$ in \textit{step 13}. In particular, the computational cost assigned to $\mathbf{p}$ is considered to be the knowledge of all its elements, that is, the computation of all possibles $P_{r(t)}$ values within a coherence block, in such a way the generation of some $r(t)$ rows comprises of drawing random numbers based on $\mathbf{p}$. This translates to the computation of $KP_{r(t)}$ values that takes a computational complexity of $2MK$ complex multiplications, which corresponds to computing, within a given coherence block, $\lVert\hat{\mathbf{g}}_{r(t)}\rVert^{2}_{2}$ for every UE and $\lVert\hat{\mathbf{G}}\rVert^{2}_{\mathrm{F}}$ only once. Now the computation of $\eta^{t}$ involves computing $\lVert\hat{\mathbf{g}}_{r(t)}\rVert^{2}_{2}$ and $\langle\hat{\mathbf{g}}_{r(t)},\mathbf{u}^{t}\rangle$ for each rKA iteration occurring within a coherence block. As the first calculation is already known by the BS hardware from the computation of $\mathbf{p}$, the computational complexity assigned to the residual computation is $M$ complex multiplications. Observe, however, that the residual term must be computed for every new rKA iteration; because of this, the successive residual computations accumulate a total complexity of $MT_{\scriptscriptstyle\mathrm{KA}}$ complex multiplications throughout the whole process described in Algorithm \ref{algthm:parlrka}.
}

{
	Another key discussion point of the UL computational complexity for the PARL rKA-based RZF scheme is related to the signal recovery problem described in \eqref{eq:signalest}. In view of the fact that each $k$th iteration of the outer loop (\textit{step 4} to \textit{step 18}) is actually computing $\mathbf{d}^{\scriptscriptstyle\mathrm{rKA}}_{k}$ and not $\mathbf{v}^{\scriptscriptstyle\mathrm{rKA}}_{k}$ in Algorithm \ref{algthm:parlrka}, it is appropriate to perform the following computations in a specific order, so as to obtain an estimation of the signals sent by UEs to the BS in a given communication instant with the most reduced complexity:
	\begin{subequations}
		\begin{equation}
		\mathbf{V}^{\scriptscriptstyle\mathrm{rKA}} = \hat{\mathbf{G}}\mathbf{D}^{\scriptscriptstyle\mathrm{rKA}}, \label{eq:cplx:UL-receveing-PARL-1}
		\end{equation}
		\vspace{-5mm}    
		\begin{equation}
		\hat{\mathbf{s}}=(\mathbf{V}^{\scriptscriptstyle \mathrm{rKA}})^{\scriptscriptstyle \mathrm{H}}\mathbf{y}.\label{eq:cplx:UL-receveing-PARL-2}
		\end{equation}	
	\end{subequations}	
	\noindent where \eqref{eq:cplx:UL-receveing-PARL-1} is computed only once and takes $MK^2$ complex multiplications, while the signal recovery in \eqref{eq:cplx:UL-receveing-PARL-2} leads to $MK$ complex multiplications and has to be computed for all $\tau_{\mathrm{ul}}$ instants.}

{
	Table \ref{tab:ULcomplexity} summarizes the overall computational complexity for the PARL rKA-based RZF receiver/combiner scheme. For the sake of comparison, it also shows the total complexity of the canonical ZF and RZF schemes based on the results disposed in \cite{Bjornson2017}. We stress that the reception column entails the computational cost associated to estimating the signals sent by the UEs on the BS side, as expressed in \eqref{eq:signalest}. Some remarks can then be made regarding the given results. The computation of the receive combining matrix for Algorithm \ref{algthm:parlrka} seems to be the most computationally light, but this value is very dependent on the number of rKA iterations $T_{\scriptscriptstyle\mathrm{rKA}}$. Here, therefore, arises the need to define a notion of convergence of Algorithm \ref{algthm:parlrka}, as done in Section \ref{sec:nume:subsec:roc:nbrofiterations} for a particular channel (NLoS dense urban) scenario of interest. The PARL rKA-based RZF scheme shows, however, the greatest complexity for reception due to the calculations described in \eqref{eq:cplx:UL-receveing-PARL-1} and \eqref{eq:cplx:UL-receveing-PARL-2}. The complexity functions obtained herein are used in Section \ref{sec:cplx:subsec:relating-ccs} to derive upper bounds for $T_{\scriptscriptstyle\mathrm{rKA}}$ based on relative computational complexities among the scheme in Algorithm \ref{algthm:parlrka} and those classical ones on which the algorithm can be based.
	\begin{table}[htp]
		\centering
		\caption{{Computational complexity per coherence interval for different receive combining schemes.}}
		\label{tab:ULcomplexity}
		\begin{tabular}{cccc}
			\hline
			\multirow{2}{*}{{\textbf{Scheme}}} & \multicolumn{2}{c}{{\textbf{Receive combining matrix}}} & {\textbf{Reception}} \\ \cline{2-4} 
			& {\textit{Multiplications}} & {\textit{Divisions}} & {\textit{Multiplications}} \\ \hline
			{ZF} & {$\frac{3K^{2}M}{2}+\frac{KM}{2}+\frac{K^3-K}{3}$} &{$K$} & {$\tau_{\mathrm{ul}}MK$} \\ \hline
			{RZF} & {$\frac{3K^{2}M}{2}+\frac{3KM}{2}+\frac{K^3-K}{3}$} & {$K$} & {$\tau_{\mathrm{ul}}MK$} \\ \hline
			\begin{tabular}[c]{@{}c@{}}{PARL rKA-based RZF}\\ {(Algorithm 1)}\end{tabular} & {$MT_{\scriptscriptstyle\mathrm{rKA}}+2MK$} & {$-$} & {$\tau_{\mathrm{ul}}MK+MK^2$} \\ \hline
		\end{tabular}
	\end{table}
}

\subsection{{Downlink Computational Cost}}
The computational complexity for the DL phase is composed of two different components: (a) the transmit precoding matrix computations, which are specified in \eqref{eq:ul-dl-duality}; and (b) the transmission stage that denotes the computational burden related to the computation of $\mathbf{x}$, {the precoded signal defined in Section \ref{sec:sub:down}}. In the special case provided when considering UL-DL duality, there is a favorable equality among the computational cost of all the schemes{\cite{Bjornson2017}, i.e., including both canonical and rKA-based schemes}. For the latter, it is important to remember that the combining matrix is already available at any given symbol period within $\tau_{\mathrm{c}}$, as demonstrated by the output values of Algorithm \ref{algthm:parlrka} (see also \eqref{eq:cplx:UL-receveing-PARL-1}). The computation of the precoding matrices consumes $MK$ complex multiplications, whereas the calculation of the precoded signals takes $\tau_{\mathrm{dl}}MK$ complex multiplications, irrespective of which precoded scheme is selected.

\subsection{{Relating Computational Complexities}}\label{sec:cplx:subsec:relating-ccs}
{
We will focus now on deriving the computational costs viewed as upper bounds of canonical and proposed schemes. Once the computational complexities for DL are the same, as shown above, it is not difficult to see from Table \ref{tab:ULcomplexity} that we can find the number of iterations $T_{\scriptscriptstyle\mathrm{rKA}}$ that makes the computational complexity of the PARL rKA-based RZF balance with that of the ZF or the RZF scheme. By equalizing the total computational complexities of the schemes in the UL, which are the sum of all the columns given in Table \ref{tab:ULcomplexity} for each scheme,from the PARL rKA-based RZF scheme to those of the considered canonical schemes, separately, and after isolating the number of rKA iterations, we get
\begin{subequations}
	\begin{equation}
	\overline{T}^{\scriptscriptstyle \mathrm{ZF}}_{\scriptscriptstyle\mathrm{rKA}} = \dfrac{K^3}{3M} + \dfrac{K^2}{2} + \dfrac{(4K-9KM)}{6M}.\label{eq:cplx:trka-zf}
	\end{equation}	
	\begin{equation}
	\overline{T}^{\scriptscriptstyle \mathrm{RZF}}_{\scriptscriptstyle\mathrm{rKA}} = \dfrac{K^3}{3M} + \dfrac{K^2}{2} + \dfrac{(4K-3KM)}{6M}, \label{eq:cplx:trka-rzf}
	\end{equation}
\end{subequations}
Those quantities tell us how many rKA iterations are necessary for Algorithm \ref{algthm:parlrka} to reach the same computational complexity of the canonical schemes considered when a system setup is defined by the pair $(M,K)$. For example, setting as an input of Algorithm \ref{algthm:parlrka}, a $T_{\scriptscriptstyle\mathrm{rKA}}$ greater than $\overline{T}^{\scriptscriptstyle \mathrm{RZF}}_{\scriptscriptstyle\mathrm{rKA}}$, given any practical values of $M$ and $K$, the computational complexity of the PARL rKA-based RZF scheme would be greater than the one placed to run the canonical RZF scheme. In practice, \eqref{eq:cplx:trka-rzf} and \eqref{eq:cplx:trka-zf} are very appealing comparison metrics that can be used to calibrate the value of $T_{\scriptscriptstyle\mathrm{rKA}}$ in Algorithm \ref{algthm:parlrka} so that it mimics the computational complexities of the canonical schemes considered. Interestingly, these quantities can be interpreted as upper bounds for the number of rKA iterations, in the means that when these values are exceeded, the proposed rKA-based scheme is no longer a computationally viable alternative in comparison with those traditionally available. A straightforward look over \eqref{eq:cplx:trka-rzf} and \eqref{eq:cplx:trka-zf} shows that other relevant upper bounds can easily be derived based on other canonical schemes, as one for the classical maximum-ratio combiner (MRC) scheme, thus making the complexity of Algorithm \ref{algthm:parlrka} very elastic depending on the intent in which the system was implemented. For example, considering a scenario where we always want to achieve a complexity with Algorithm \ref{algthm:parlrka} below to the one provided by the ZF scheme for whatever performance is achieved. To accomplish this we simply set a number of rKA iterations as $\{T_{\scriptscriptstyle\mathrm{rKA}}\in\mathbb{N}^{+}|T_{\scriptscriptstyle\mathrm{rKA}}<\overline{T}^{\scriptscriptstyle \mathrm{ZF}}_{\scriptscriptstyle\mathrm{rKA}}\}$. Note that the bounds are only dependent upon the system scale definer pair $(M,K)$. To draw some general quantitative conclusions on the performance--computational complexity trade-off of the PARL rKA-based RZF and classical RZF, Section \ref{sec:nume:subsec:complexity} develops the notion of convergence based on an average value of $T_{\scriptscriptstyle\mathrm{rKA}}$, which is derived in Section \ref{sec:nume:subsec:roc:nbrofiterations}, with the upper bound in \eqref{eq:cplx:trka-rzf}.}

\section{Numerical Results}\label{sec:results}
Let us now consider a square-network layout to gather some quantitative results on the applicability of the Kaczmarz methodology in a more practical M-MIMO scenario. A single-cell case adopting the parameters exhibited in Table \ref{tab:simulation-parameters} is considered throughout the simulations. The cell covers an area of $0.25\times0.25$ km\textsuperscript{2} with a {center located} BS equipped with $M$ antennas spatially multiplexing $K$ UEs. The UEs are uniformly distributed inside the cell area with a minimum distance of $35$ m to the BS. The average long-term fading denoting the pathloss effect for UE $k$ is evaluated as: $\beta_{k}=\Gamma-10\alpha\log_{10}(d_{k})$ [dB]; where $d_{k}$ is the Euclidean distance between UE $k$ and the BS, $\Gamma$ is the pathloss constant term at a reference distance of $1$ m, and $\alpha$ is the pathloss exponent. Since otherwise affirmed, the ordinary values of $\Gamma=-35.3$ dB and $\alpha=3.76$ are used based on a \textit{NLoS dense urban scenario} (features like cell size and large scale effects are very related to the type of scenario considered) \cite{Marzetta2016,Bjornson2017}. {By joining these values and those presented in Table \ref{tab:simulation-parameters}, the UE close to the center of the cell or center-UE has an average SNR of $17.63$ dB, while the UE located on the edge of the cell or edge-UE possesses $-14.47$ dB; one should note that the average disregards the shadowing term that fulfills some variability of the specified values. In addition to that, three scenarios for the loading factor were assumed: the lowly loaded scenario with ${K}/{M}=0.1$, the moderately loaded with ${K}/{M}=0.3$, and the highly loaded with ${K}/{M}=0.5$.}

\begin{table}[htp]
	\centering
	\caption{Parameters considering a square-network layout for a single-cell M-MIMO scenario {under a dense urban channel condition}.}
	\label{tab:simulation-parameters}
	\begin{tabular}{ll}
		\toprule
		\textbf{Parameter} & \textbf{Value} \\ \hline
		Single-cell area & $0.25\times0.25$ km\textsuperscript{2} \\
		BS antennas ($M$) & $100$ \\
		\# UEs  ($K$) & $10;30;50$ \\
		{Loading factor ($K/M$)} & {$10\%;30\%;50\%$} \\ 
		Bandwidth & $20$ MHz \\ 
		Receiver noise power & $-91$ dBm \\ 
		UL transmit power per UE & $20$ dBm \\ 
		Antenna correlation factor ($r $)& $0.5$\\
		Coherence interval ($\tau_{\mathrm{c}}$) & $200$ symbols \\ 
		Pilot length ($\tau_{\mathrm{p}}$) & $K$ symbols \\ \hline
		Pathloss exponent ($\alpha$)& $2.00; 3.76; 4.00$ \\
		Pathloss constant term ($\Gamma$)& $-35.3$ dB \\
		Shadowing standard deviation ($\sigma$) & $4$ dB  \\ 
		{SNR range} & {$[-14.47,17.63]$ dB} \\ \hline
		Channel estimation method & LS and MMSE \\
		\bottomrule
	\end{tabular}
\end{table}

Canonically, the covariance matrix model for uncorrelated Rayleigh fading is
\begin{equation}
\mathbf{R}_{k}=\beta_{k}10^{f/10}\mathbf{I}_{M},
\label{eq:spatial-corr-model:uncorrelated}
\end{equation}
where $f \sim \mathcal{N}({0},{\sigma^2})$ stands for the {long-term} fading (shadowing) with zero mean and standard deviation $\sigma$, which is constant to all elements of the array of antennas \cite{Marzetta2010,Marzetta2016}. Motivated by the results obtained in \cite{Bjornson2018a}, it is also assumed a spatial correlation model composed of the exponential correlation model for a uniform linear array (ULA) with long-term fading variations over the array \cite{Loyka2001,Gao2015,Gao2015a}. The exponential correlation model stands for spatial antenna correlation, while shadowing fluctuations represent the power variability across antenna elements caused by the disordered distribution of the scatters in the environment. {One should be emphasized that this last effect is empirically justified by the conclusions acquired in \cite{Gao2015,Gao2015a}}. That being said, the $(m,n)$th element of the covariance matrix for the $k$th UE under this combined spatially correlated channel is defined as \cite[p. 11, 1\textsuperscript{st} column]{Bjornson2018a}
\begin{equation}
\left[\mathbf{R}_{k}\right]_{m,n} = \beta_{k} r^{\lvert{n-m}\vert} e^{i(n-m)\theta_{{k}}} 10^{{(f_{m} + f_{n})/20}},
\label{eq:spatial-corr-model:correlated}
\end{equation}
where $m,n\in\{1,\dots,M\}$, $r \in [0,1]$ is the antenna correlation factor, and $\theta_{k}$ is the angle-of-arrival of the $k$th UE {signal}, which is uniformly distributed in $[-\pi,+\pi)$ when considering an horizontal ULA. {Besides,} $f_1,\dots,f_{m},\dots,f_{M} \sim \mathcal{N}({0},{\sigma^2})$ are i.i.d. random variables that {form} the random fluctuations of the long-term fading with zero mean and a standard deviation $\sigma$. We take this moment to define the moderate spatial correlation notion for this model as the scenario in which $r=0.5$ and $\sigma=4$ dB, as considered in \cite{Bjornson2018a}. This moderate spatial correlation setting is often used in the simulations.

As we are considering a single-cell {where UEs are assigned with mutually orthogonal pilots}, the channel estimation process is observed to be only noise limited. Some important observations about the covariance matrices of the channel estimates can then be made to further clarify the evaluations carried out in this section. One of the effects of channel estimation over the eigenvalues of the true channel's covariance matrix is generally the reduction of their magnitudes, once all estimate is comprised of error \cite{Bjornson2017}. This error is seen higher for the LS estimator in comparison to the MMSE estimator, where the difference is due to the use or not of prior statistical information. Another known key result is that the stronger the eigenvalue associated with a given eigendirection, the easier it becomes the channel estimation related to this spatial direction \cite{Bjornson2017}. Note that a higher eigenvalue actually translates in {increasing SNR}, which explains the best channel estimates. Since spatial correlation impacts the eigenvalues of the channel covariance matrix, where few of them even carry a large fraction of the power, as shown in \cite{Bjornson2018a,doi:10.1002/ett.3563} for the considered model, spatially correlated channels are therefore susceptible to a better estimate. For a strong spatial correlation, less eigenvalues have considerable power, further enhancing the channel estimation related to these spatial directions that are the most significant for establishing the link between the BS antennas and a UE. When the structure of the channel is not duly exploited, however, as in the case of the LS estimator, the scaling of the channel estimate is difficult to be correctly obtained, and thus only the channel direction can be estimated satisfactorily \cite{Bjornson2017}. {We stress that these concepts will help to understand the behavior of some curves obtained below.}

\subsection{Impacts Detrimental to Sample Probability}\label{sec:nume:subsec:rowprob}
As can be seen from \textit{step 11} of Algorithm 1, the sample probability quantifies the relative power of a given UE in relation to the system. $P_{r(t)}$ is therefore mainly related to the power of the estimated channels, whereby in our model it is also accounting for the pathloss and shadowing effects, {as can be seen in \eqref{eq:spatial-corr-model:uncorrelated} and \eqref{eq:spatial-corr-model:correlated}}. One can then assume that this value is very sensitive to long/short-term fading, channel estimation procedure, and even impacted by spatial correlation; being some of these effects illustrated further on.

\begin{figure}[htbp]
	\centering
	\includegraphics[width=.75\linewidth]{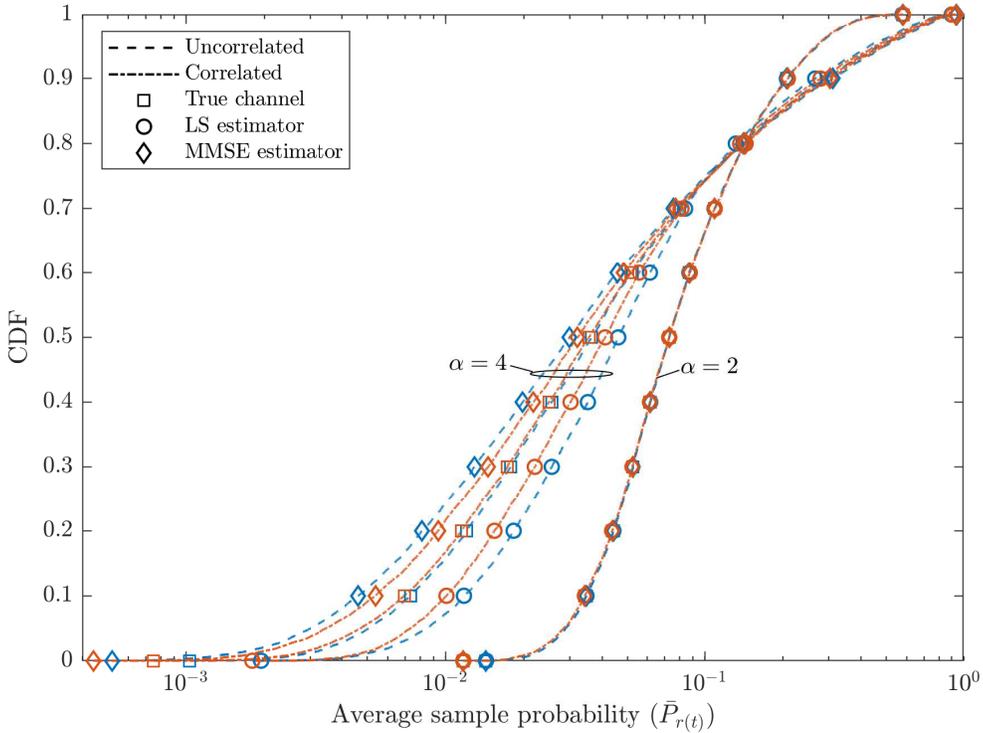}
	\caption{{CDFs of the sample probability $P_{r(t)}$ defined in \textit{step 10} of Algorithm \ref{algthm:parlrka} and averaged over independent small-scale fading realizations with ${K}/{M}=0.1$, $r=0.5$, and $\sigma=4$ dB. The probabilities were also generated for two different values of $\alpha$ and when applying the LS/MMSE channel estimators. The true channel case indicates the perfect knowledge of CSI in the BS.}}
	\label{fig:figure1}
\end{figure}

By treating $P_{r(t)}$ as a random variable, Fig. \ref{fig:figure1} gives the cumulative distribution functions (CDFs) of the average sample probability, as specified in \textit{step 11} of Algorithm 1, {for different values of $\alpha$.} It is clear to see that $P_{r(t)}$ is strongly influenced by pathloss, since the CDF curves are shifted to the left as $\alpha$ grows. This result corroborates the well-known fact that edge-UEs are doomed, naturally, to have a minor level of power than the center-UEs; ergo, the former also own lower values of $P_{r(t)}$ that become more and more distant from those obtained for center-UEs as $\alpha$ increases. Note that the channel {estimator} has a reasonable impact as well. The curves for the MMSE estimator are more to the left for $\alpha = 4$, thus showing that this {method} can obtain smaller values of $P_r(t)$ in comparison to LS. This is a direct consequence of the fact that the MMSE channel estimator works well in estimating the channels of edge-UEs. When using the MMSE, there is a small deviation to the right of the spatially correlated curves in relation to the uncorrelated case that can be explained by the improvement of the channel estimation under a moderate level of spatial correlation. Recall that this positive result is only partly obtained for the LS estimator and hence the negative effect of setting a reasonable estimate for the power scaling factor becomes more evident in this case; {thus justifying the different behavior of the curves in the LS case in comparison to that observed for the MMSE channel estimator}. In summary, one can conclude that edge-UEs may have very low values of $P_r(t)$ when compared to center-UEs in scenarios of practical interest; {proving in this way our hypothesis raised in Section \ref{sec:KA:precoding/decoding}}. This fact hinders the {stochastic selection} of the equations attached to edge-UEs, as can be seen in \textit{step 11} of Algorithm \ref{algthm:parlrka}, making it impossible to perform the rKA’s update schedule correctly{, when step 9 is considered to be nonexistent}. As a result, the genuine rKA without any modification applied to this more practical context has a restricted performance.

\begin{figure}[htp]
	\centering
	\includegraphics[width=.95\linewidth]{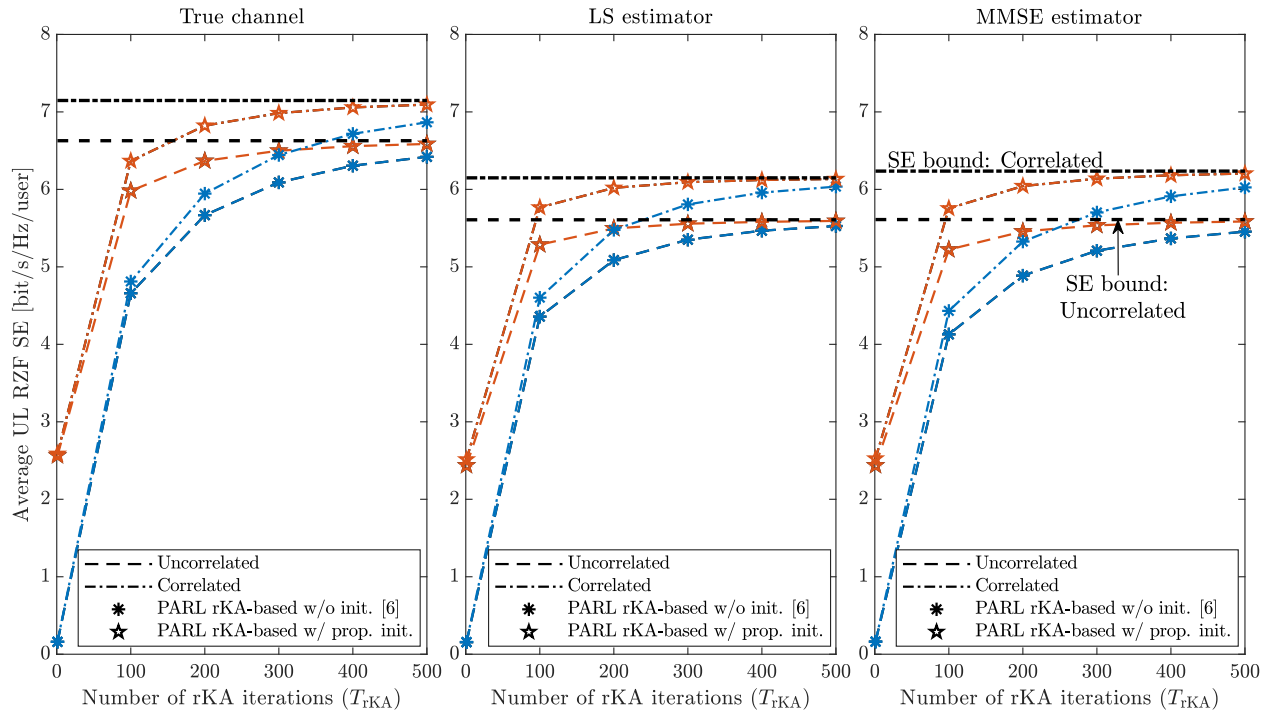}
	\caption{{Average UL SE per UE as a function of $T_{\scriptscriptstyle\mathrm{rKA}}$ for the canonical RZF combining and its emulations performed by the PARL rKA-based schemes. The first PARL algorithm follows the ideas of \cite{Boroujerdi2018b} and does not have an initialization (abbreviated as init.) procedure; the other proposed (abbreviated as prop.) by us in Algorithm \ref{algthm:parlrka} has a hybrid initialization that guarantees the convergence of all receive combining vectors. Three different channel estimates are considered: (a) idealistic case where the true channel is known at the receiver, (b) LS channel estimates, and (c) MMSE channel estimates. The scenario under evaluation consists in ${K}/{M}=0.1$, $r=0.5$, and $\sigma=4$ dB.}}
	\label{fig:figure2}
\end{figure}

Fig. \ref{fig:figure2} exhibits the average UL SE per UE as a function of the number of rKA iterations and different channel estimates for the equivalent PARL rKA-based RZF scheme conceived in \cite{Boroujerdi2018b} {and the proposed/modified by us using the hybrid initialization method described in Algorithm \ref{algthm:parlrka}}. We stress that the former does not deem with the practical effect visualized in Fig. \ref{fig:figure1}, whereas our proposed methodology takes into account these harmful differences observed in $P_r(t)$ for center- and edge-UEs through a fair, forced (hybrid) initialization process. It is clearly visible that the proposed scheme has surpassed its correspondent one \cite{Boroujerdi2018b}, since our proposal always converges faster {to the reference bounds given by the average RZF SE, where it is important to note that the SEs are being computed using \eqref{eq:sekul}}. Besides that, the average UL SEs obtained are in agreement with the results of \cite{Bjornson2018a}, where the authors have adopted the same spatial correlation model but for the DL phase. The superiority in performance of spatially correlated channels stems from the fact that the spatial correlation arising from antenna correlation and unequal contribution of the antennas is, on average, corroborative to the SE bound defined in \eqref{eq:sekul}. {One of the consequences behind it are the effects of spatially correlated channels on the channel estimation, as explained at the beginning of this section.} In addition, it is important to keep in mind that the use-and-forget technique, used to define the {ergodic capacity in} \eqref{eq:sekul}, is very sensitive to the assurance of channel hardening, which, in turn, is dramatically impacted by the spatial correlation phenomenon, as reported in \cite{Bjornson2017,doi:10.1002/ett.3563}.

In conclusion, long-term channel effects have a noticeable impact on the performance of rKA, and their concern in the design of the rKA-based schemes applied to M-MIMO allows better results. We emphasize that the improvements achieved by the proposed algorithm does not imply any increase in computational complexity compared to that based on the original idea \cite{Boroujerdi2018b}, {where one can observe that} our hybrid proposal also appeals because of its broader applicability. {Some other inferences can be drawn from the results obtained here.} As the fact that when the true channel is known at the receiver side, it seems necessary more rKA iterations than the case of MMSE channel estimates to follow the boundaries of SE, which in its turn needs more rKA iterations than the worst channel estimation case constituted by LS channel estimates.

\subsection{Rate of Convergence: Number of Iterations Required}\label{sec:nume:subsec:roc:nbrofiterations}
{Our concern now is to define a notion of convergence based on the number of rKA iterations required to reach a predefined error limit by considering the canonical, average RZF SE per UE as the true value. This will give us more insights into how the algorithm converges and will also help to establish a way to evaluate in the sequel its computational complexity. The gap in percentage between the average SEs per UE of the canonical RZF scheme and of its emulation performed through Algorithm \ref{algthm:parlrka} is shown in Fig. \ref{fig:figure3}, as a function of the average number of rKA iterations for {uncorrelated and moderately correlated channels}, different channel estimates, and different loading factor values. Starting with the last, the figure shows that a better convergence is obtained for smaller values of ${K}/{M}$, as expected from earliest ideas conceived from Remark \ref{rmk:parl2}. Increasing ${K}/{M}$ expands the solution subspace of the SLE which in its turn increases the interference among UEs' solutions. Fig. \ref{fig:figure3} embraces the perception that the better the CSI at the receiver, the worse the algorithm convergence. The reason why this occurs comes from Remark \ref{rmk:parl1} in which it is possible to observe that as better is the channel estimation, the better is also the characterization of power for the estimated channels; {consequently, spawning disturbances on the division of a UE power and the power of the system}. {Estimating the channels with a more naive approach, therefore, improves the convergence of Algorithm \ref{algthm:parlrka}. One should bear in mind, however, that the achievable average SE is also diminished.} We also note that better convergence results are obtained for the uncorrelated case in comparison to moderately, spatially correlated channels. To explain this, we must again return to Remark \ref{rmk:parl1}, where we observe that the average gain is severely impacted by $\lambda_{\min}(\hat{\mathbf{G}}^{\scriptscriptstyle \mathrm{H}}\hat{\mathbf{G}})$, which is sharply disturbed by the spatial correlation. The penalty brought by the considered spatial correlation model is best detailed in the next section. 
	
\begin{figure}[htp]
	\centering
	\includegraphics[width=.75\linewidth]{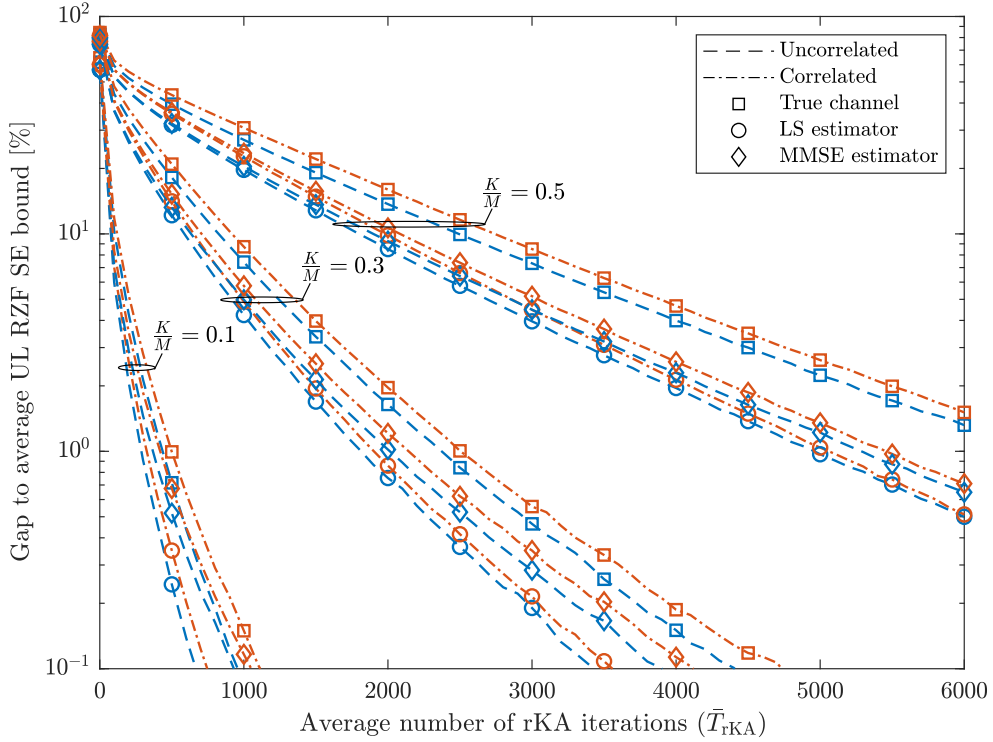}
	\caption{Percentage gap between the average UL SEs per UE of the canonical RZF scheme and its correspondent emulation performed through the proposed PARL rKA-based scheme (Algorithm \ref{algthm:parlrka}) as a function of the average number of rKA iterations with $M=100$, $r=0.5$, and $\sigma=4$ dB.}
	\label{fig:figure3}
\end{figure}

Table \ref{tab:average-number-of-iterations} gives the average number of rKA iterations $\bar{T}_{\scriptscriptstyle\mathrm{rKA}}$ to reach two definitions of convergence chosen to obtain average errors from the canonical performance, which are: (a) less than $10\%$ and (b) less than $1\%$. {For example, if the average UL SE per UE of the RZF was $2$ bit/s/Hz when applied to a scenario where ${K}/{M}=0.1$ and the channels were uncorrelated, a BS running the proposed rKA-based scheme would run on average $93$ times the inner loop (\textit{step 7} to \textit{step 16}) of Algorithm \ref{algthm:parlrka} to obtain a UL SE per UE of $1.8$ bit/s/Hz, whereas $293$ iterations are needed on average to get a UL SE per UE of $1.98$ bit/s/Hz. Recalling that the BS is equipped with simple, cheap hardware, we observe that the most realistic implementation scenario is when the LS\footnote{The MMSE estimator has been used so far, only to provide insights as the ideal case of channel estimation. It goes, however, against the placed underlying paradigm for the BS that is unable to estimate the second moment of channel responses.} channel estimator is employed on the BS side due to its computational simplicity \cite{Bjornson2018a,Gao2015}.} {This explains why we only consider quantities related to the LS channel estimator in Table \ref{tab:average-number-of-iterations}}. {Moreover,} these values were obtained through a linear interpolation between the points of the curves in Fig. \ref{fig:figure3}, wherein an interval of $100$ average number of rKA iterations was adopted between neighboring points.

\begin{table}[htp]
	\centering
	\caption{Average number of rKA iterations $\bar{T}_{\scriptscriptstyle\mathrm{rKA}}$ required to the proposed PARL rKA-based RZF scheme can achieve a defined error tolerance of the attainable SE for LS channel estimation.}
	\label{tab:average-number-of-iterations}
	\begin{tabular}{ccccc}
		\toprule
		\multirow{3}{*}{$\frac{K}{M}$} & \multicolumn{4}{c}{$\bar{T}_{\scriptscriptstyle\mathrm{rKA}}$} \\ \cline{2-5} 
		& \multicolumn{2}{c}{\textbf{Error bound of 10\%}} & \multicolumn{2}{c}{\textbf{Error bound of 1\%}} \\ \cline{2-5} 
		& \textit{Uncorr./Corr.} & \textit{Ratio\tnote{$\dagger$} {[}\%{]}} & \textit{Uncorr./Corr.} & \textit{Ratio {[}\%{]}} \\ \hline
		$0.1$ & $93/95$ & $-2.15$ & $293/333$ & $-13.65$ \\ \hline
		$0.3$ & $589/655$ & $-11.21$ & $1815/1903$ & $-4.85$ \\ \hline
		$0.5$ & $1799/1983$ & $-10.23$ & $4960/5062$ & $-2.06$ \\ \bottomrule
	\end{tabular}%
	\begin{tablenotes}
		\item[$\dagger$] Ratio: difference in percentage of the correlated case in relation to the uncorrelated case.
	\end{tablenotes}
\end{table}	

\subsection{Rate of Convergence: Spatial Correlation Effects}
{We now set out to better understand how the considered spatial correlation model impacts the algorithm convergence. For this purpose, results of Fig. \ref{fig:figure3} were then extended\footnote{{For the sake of clarity, we have omitted the true channel case; but the obtained responses are still valid under this context.}}, as can be seen in Fig. \ref{fig:figure4}, by considering both $r$ and $\sigma$ dimensions for $\frac{K}{M}=0.1$. When $r$ is varying and $\sigma=0$ dB, {the general behavior tells that the antenna spatial correlation is detrimental to the performance of the PARL rKA-based RZF scheme just in highly antenna correlated channels ($r > 0.7$)}. The same result is completely obtained for all the region of the opposite case wherein only the spatial correlation arising from the environment disorder is considered; confirming thus the findings obtained previously. Fig. \ref{fig:figure3} also corroborates the expansion of the idea that better channel estimates deteriorate the convergence speed of Algorithm \ref{algthm:parlrka}. Again, the reason for expecting performance variations from spatial correlation effects proceeds, strictly speaking, from the relevant changes in the eigenstructure of $\hat{\mathbf{G}}^{\scriptscriptstyle \mathrm{H}}\hat{\mathbf{G}}$ inferred by the degree of spatial correlation that depends on the values of $r$ and $\sigma$ (see the bounds in Remark \ref{rmk:parl1}).}
\begin{figure}[htp]	
	\centering
	\subfloat[{Percentage gap as a function of $\bar{T}_{\scriptscriptstyle\mathrm{rKA}}$ and $r$ with $\sigma=0$ dB.}]{%
		\includegraphics[width=.49\textwidth]{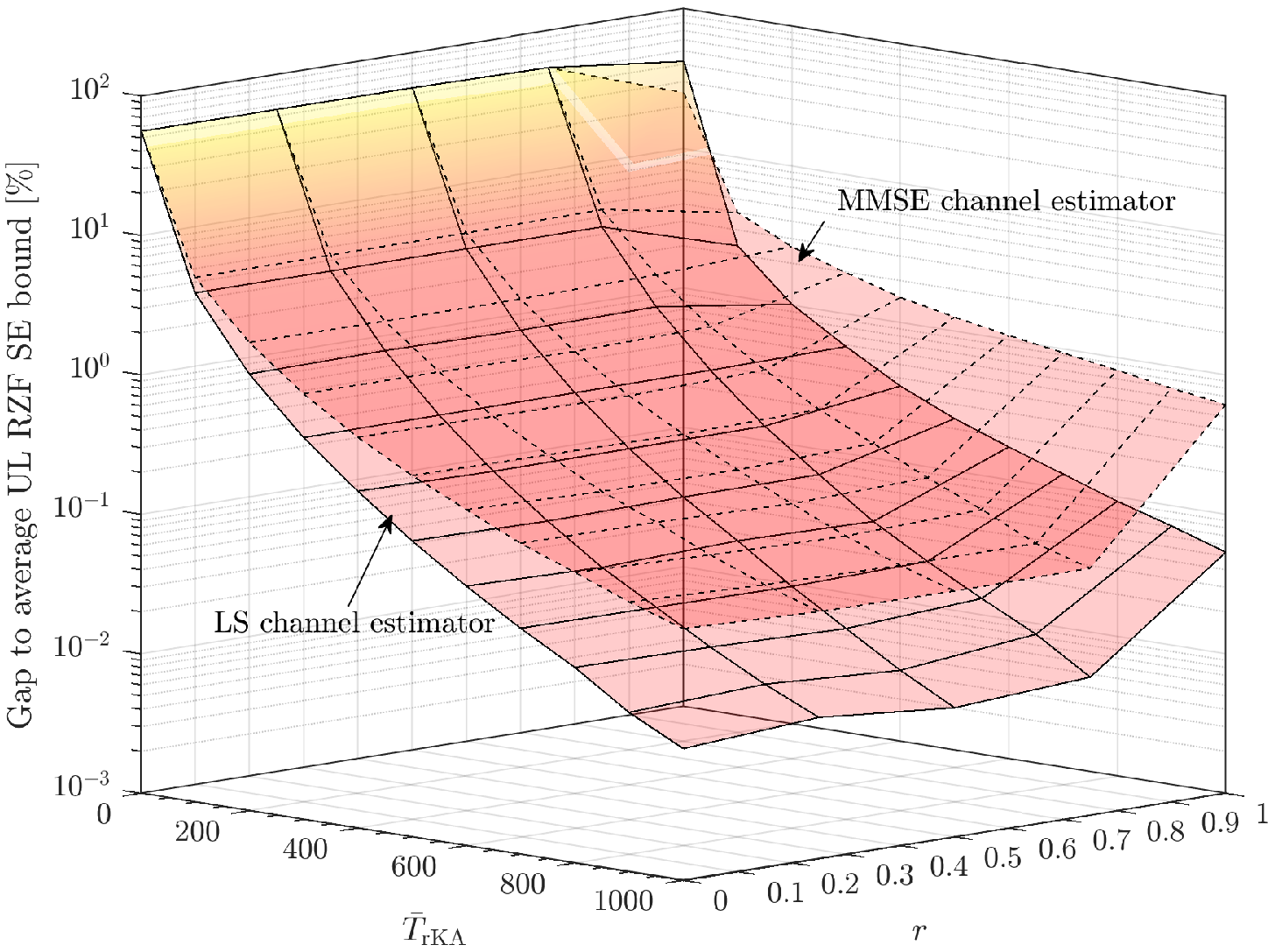}%
		\label{fig:figure4a}
	}
	\subfloat[{Percentage gap as a function of $\bar{T}_{\scriptscriptstyle\mathrm{rKA}}$ and $\sigma$ with $r=0$.}]{%
		\includegraphics[width=.49\textwidth]{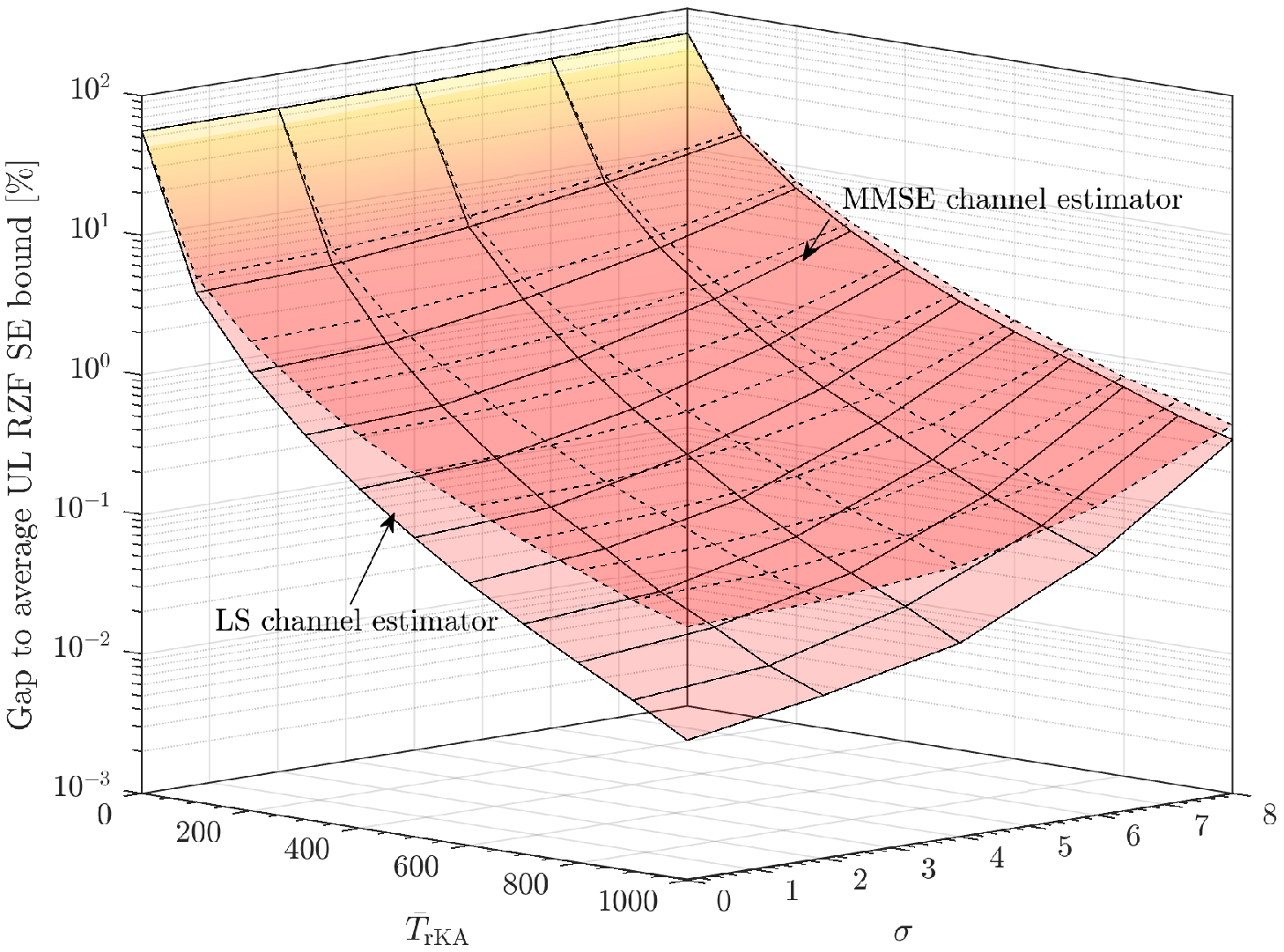}%
		\label{fig:figure4b}
	}
	
	\caption{{Percentage gap between the average UL SEs of the canonical RZF scheme and its correspondent emulation performed through the proposed PARL rKA-based scheme (Algorithm \ref{algthm:parlrka}) as a function of the average number of rKA iterations ($\bar{T}_{\scriptscriptstyle\mathrm{rKA}}$) with ${K}/{M}=0.1$. The graph also considers the variation of the gap in: (a) the antenna correlation $r$ dimension and (b) the shadowing standard deviation $\sigma$ dimension}.}
	\label{fig:figure4}
\end{figure}

\subsection{{Performance--Complexity Trade-Off}}\label{sec:nume:subsec:complexity}
{
	Given the average number of rKA iterations discussed in Section \ref{sec:nume:subsec:roc:nbrofiterations} for the aforementioned scenario of interest (typical NLoS dense urban scenario), which embodies only LS channel estimates, we are now in position to analyze the trade-off considering both the performance and computational complexity between the use of the PARL rKA-based RZF and the canonical RZF schemes to obtain the receive combining matrix in an M-MIMO communication setup. In order to get a more extensive analysis, we compare the upper bound defined in \eqref{eq:cplx:trka-rzf} to the average number of rKA iterations $\bar{T}_{\scriptscriptstyle \mathrm{rKA}}$ needed to attain an average SE per UE, our considered performance metric of 10\% or 1\% less than that given by the canonical RZF scheme, which we denote as $\mathrm{SE}^{\scriptscriptstyle\mathrm{rKA}}_{\scriptscriptstyle 10\%}$ and $\mathrm{SE}^{\scriptscriptstyle\mathrm{rKA}}_{\scriptscriptstyle 1\%}$, respectively. Fig. \ref{fig:figure5} illustrates the behavior of $\overline{T}^{\scriptscriptstyle\mathrm{RZF}}_{\scriptscriptstyle \mathrm{rKA}}$ when the number of BS antennas is varying. At the $m$th point in the figure, $K$ was chosen in such a way that the loading factor equals one of the values in $\{0.1,0.3,0.5\}$, yielding three curves showing the scaling of $\overline{T}^{\scriptscriptstyle \mathrm{RZF}}_{\scriptscriptstyle \mathrm{rKA}}$. To include the performance results, the values in Table \ref{tab:average-number-of-iterations} were approximated downwards to an integer value of $M$ that results in a value less or equal than the ones disposed in the table. This process results in the points marked with circles for uncorrelated channels, and asterisk for moderately, spatially correlated channels with $r=0.5$ and $\sigma = 4$ dB. Grouping this points with respect to performance and linking them, we get trade-off lines that are proportional to the true trade-off between the performance of the PARL rKA-based algorithm, denoted by $\mathrm{SE}^{\scriptscriptstyle\mathrm{rKA}}_{\scriptscriptstyle 10\%}$ and $\mathrm{SE}^{\scriptscriptstyle\mathrm{rKA}}_{\scriptscriptstyle 1\%}$, and the computational complexity, represented by $\mathrm{CC}^{\scriptscriptstyle \mathrm{rKA}}_{\scriptscriptstyle\mathrm{RZF}}$, of the proposed scheme in relation to the canonical RZF. The latter, of course, taking into account the complexity imposed by different system setups delimited by $(M,K)$ pairs. 
}

{
	In Fig. 5 the arrows indicate that above the trade-off lines, the PARL rKA way of implementing the RZF is more computationally light than its classical implementation way. When using the trade-off lines, therefore, one is accepting the loss of performance obtained using the proposed scheme. In view of this observation, we can hence read this graph as: for $M>138$ and an expected performance loss of 10\%, Algorithm \ref{algthm:parlrka} is placed as the most attractive computationally on average, if $T_{\scriptscriptstyle \mathrm{rKA}}$ is selected accordingly to Table \ref{tab:average-number-of-iterations} and whether the channel is correlated or not; this also holds for an expected performance loss of 1\% when $M>254$.
}

{
	From Fig. \ref{fig:figure5}, it is also possible to see that the fraction between any $\overline{T}^{\scriptscriptstyle \mathrm{RZF}}_{\scriptscriptstyle \mathrm{rKA}}$ and $\bar{T}_{\scriptscriptstyle \mathrm{rKA}}$ gives us the proportion of how many times the proposed algorithm is less or more complex taking as reference the complexity of the emulated (RZF) scheme. Assuming an underlying scenario, for example, where we have a BS equipped with a ULA of $M=200$ antennas and is serving $K=100$ UEs, these linked by channels being considered here as correlated, i.e., $K/M=0.5$, if the BS is running the LS channel estimator and the PARL rKA-based scheme, and the services delivered by it accept a performance loss of 10\% (operating point marked by the black diamond in the figure), the BS hardware would be saving $6617/1953\approx3.39$ times less processing than the necessary to run the canonical RZF.
}

{
	In summary, for the Kaczmarz methodology be a reasonable answer to low complexity signal processing schemes, the application scenario in which the M-MIMO communication system will be installed must be widely known. To put it differently, since they can clearly change the values acquired in Table \ref{tab:average-number-of-iterations}, parameters like those involved with the pathloss ($\alpha$, $\Gamma$), system characteristics (cell area, $\mathrm{SNR}$), and other physical and stochastic-modeled phenomena (spatial correlation, $\sigma$) must be previously measured/estimated by the network designer. In this case, one can investigate if the average number of rKA iterations necessary for suitable convergence, given a system scale pair $(M,K)$ in which the system will commonly operates, is really achievable to fulfill our primary goal of obtaining a reduced processing cost BS. This is placed due to the fact that the computational complexity saving of the proposed scheme depends tremendously on $\bar{T}_{\scriptscriptstyle\mathrm{rKA}}$, and can be obtained, as seen from Fig. \ref{fig:figure5}, at the cost of deploying more BS antennas, which is expensive in economic terms and from the point of view of the complexity of the system, or reducing the performance (we consider losses of 1\% and 10\% as good evaluation margins) of the proposed algorithm, which, if placed too far from the performance bound, can considerably lessen the relevance and functionality of the communication system.
} 
\begin{figure}[htp]
	\centerline{\includegraphics[width=.75\linewidth]{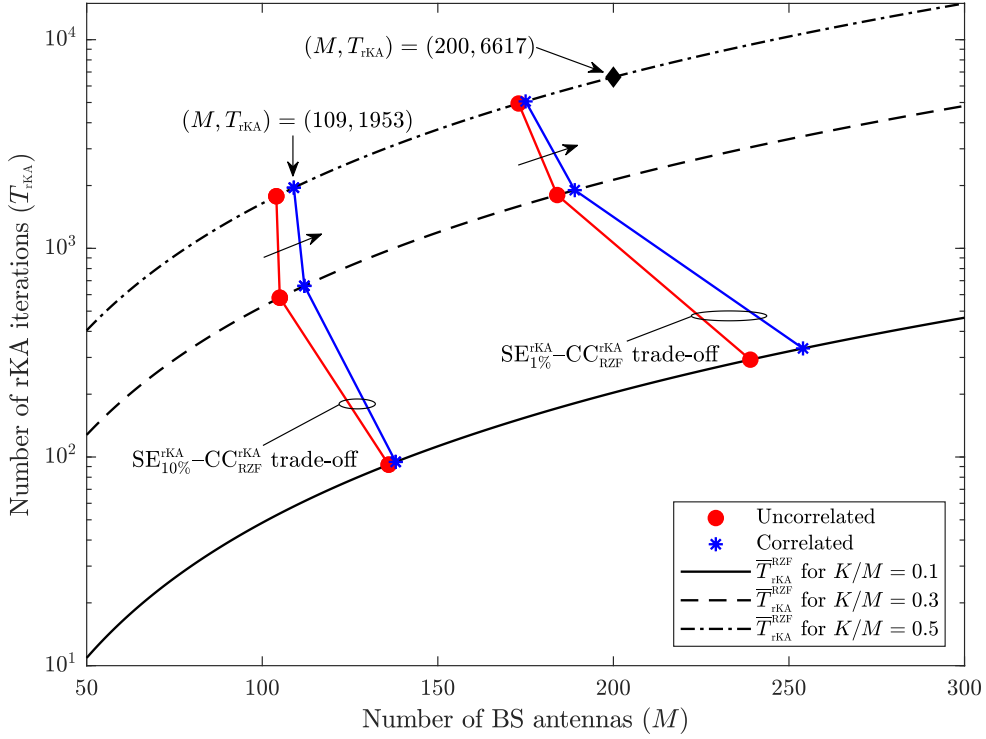}}
	\caption{{Number of rKA iterations $T_{\scriptscriptstyle\mathrm{rKA}}$ in relation to the number of BS antennas $M$, where the curves and values follow the upper bound for $T_{\scriptscriptstyle\mathrm{rKA}}$ defined in \eqref{eq:cplx:trka-rzf} and the quantities discussed in Table \ref{tab:average-number-of-iterations}. The M-MIMO communication scenario evaluated here considers uncorrelated and correlated Rayleigh fading, LS channel estimator, and loading factors $K/M\in\{0.1,0.3,0.5\}$. As the value of $M$ is defined by the abscissa axis, for each $m$th number of BS antennas point, $K$ was computed in such way that $K/M$ satisfies the loading factor values considered in each curve. Demanding better performance of Algorithm \eqref{algthm:parlrka} causes the trade-off frontiers to go to the right, increasing the number of antennas $M$ required for the proposed algorithm becomes appealing.}} 
	\label{fig:figure5}
\end{figure}

\section{Conclusions}\label{sec:conc}

In this paper, we revisited the solution to the problem of computational relaxation for canonical M-MIMO signal processing schemes based on the Kaczmarz methodology for solving a set of linear equations. {In particular, we generalized the application of rKA to obtain the combining/precoding matrices in scenarios where the UEs are allowed to have higher mobility. We have shown that practical long-term effects of wireless channels degrade the performance of the rKA algorithm. To counter this harmful effect, a hybrid, forced initialization approach was proposed and proved theoretically and numerically to surpass, regarding performance and robustness, the raw strategy available in the literature. 
	We also studied the computational cost of the proposed scheme in relation to the classical ones, where upper bounds for the number of iterations were given targeting to set a computationally efficient strategy. Numerical results based on a NLoS dense urban scenario could then be used to drawn important insights about the performance--complexity trade-off of the proposed rKA-based receive combining algorithm. On one hand, we found that the rKA can achieve the complexity gains needed to enable cheaper M-MIMO BSs with good performance results (90\% and 99\% of those obtained with canonical schemes), thus fulfilling the objectives issued even with the limitations imposed by the system and channel scenarios considered. On the other hand, these favorable results are not always viable in practice, since they rely heavily on the characteristics of the scenario considered, including furthermore the physical conditions of the environment, and, sometimes, one needs to pay large penalties on performance (trying to lower $M$) for complexity competitive reasons.} Still, the theoretical characterization of the convergence function with respect to the different system parameters or metrics would be interesting to better understand the applicability of the Kaczmarz algorithm for M-MIMO. This is indeed left as a potential avenue for future research.

\appendix

\section{Appendix A: Useful Results}
The following definitions are considering an SLE in its canonical form $\mathbf{A}\mathbf{x}=\mathbf{b}$, where $\mathbf{A}\in\mathbb{C}^{m\times{n}}$ is the matrix of the constant coefficients, $\mathbf{x}\in\mathbb{C}^{n}$ is the vector of unknowns, and $\mathbf{b}\in\mathbb{C}^{m}$ is the vector of known offset coefficients. The subspace generated by the columns of $\mathbf{A}$ is denoted as $\mathcal{X}\subset \mathbb{C}^{n}$. This SLE is solved via rKA in which the random row selection of $\mathbf{A}$ is denoted as a random variable $R\in\{1,2,\dots,m\}$ and a specific random row picked at iteration $t$ is $z\in\{1,2,\dots,m\}$. It is assumed that $z$ has a generic sample probability given as $P_{z}\in[0,1]$, whereby it is possible to define a probability vector as $\mathbf{p} = (P_{1},\dots,P_{m})\in\mathbb{R}^{m}_{+}$. A remarkable property is that $\sum_{z=1}^{m}P_{z}=1$.

\theoremstyle{definition}
\begin{definition}[Random Rank-1 Projection Operator]
	The random rank-1 projection operator of a random chosen row $R$ is \cite{Boroujerdi2018b}:  $\mathcal{P}_{R}:={(\mathbf{a}_{R}\mathbf{a}^{\scriptscriptstyle \mathrm{H}}_{R})}/{\lVert\mathbf{a}_{R}\rVert^{2}_{2}}$ , where $\mathcal{P}_{R}$ is an $n\times{n}$ matrix.
	\label{apx:def:projoperator}
\end{definition}
\begin{definition}[Average Random Rank-1 Projection Operator]
	From Definition \ref{apx:def:projoperator}, it is possible to specify the average projection operator as \cite{Boroujerdi2018b}: 	$\bar{\mathcal{P}}:=\mathbb{E}\{\mathcal{P}_{R}\}=\sum_{z=1}^{m}P_{z}({\mathbf{a}_{z}\mathbf{a}^{\scriptscriptstyle \mathrm{H}}_{z}})/{\lVert\mathbf{a}_{z}\rVert^{2}_{2}}$, where the expectation is w.r.t. $R$. Thus, $\bar{\mathcal{P}}$ is an $n\times{n}$ positive-semidefinite matrix with $\mathrm{tr}({\bar{\mathcal{P}}})=\sum_{z=1}^{m}P_{z}=1$.
	\label{apx:def:avgprojoperator}
\end{definition}
\begin{definition}[Average Gain]
	The average gain is a suitable metric that measures the exploitation provided by the average projection operator over $\mathcal{X}$ seeking to obtain $\mathbf{x}$, given as \cite{Boroujerdi2018b}:
	$\kappa_{\mathcal{X}}(\mathbf{A},\mathbf{p}):=\min_{\boldsymbol{\vartheta}\in\mathcal{X},\boldsymbol{\vartheta}\neq0}({\boldsymbol{\vartheta}^{\scriptscriptstyle \mathrm{H}}\bar{\mathcal{P}}\boldsymbol{\vartheta}}) / {\lVert\boldsymbol{\vartheta}\rVert^{2}_{2}}$, where one observes that $\kappa_{\mathcal{X}}$ relies on $\mathbf{A}$ and $\mathbf{p}$. Besides, it is possible to infer that $\kappa_{\mathcal{X}}(\mathbf{A},\mathbf{p})\in[0,\frac{1}{\min\{m,n\}}]$. This last result is a consequence of
	$\mathrm{tr}(\bar{\mathcal{P}})=1$, and of the subspace $\mathcal{X}$ has its dimension generally determined by $\min\{m,n\}$.
	\label{apx:def:avggain}
\end{definition}

\section{Appendix B: Proof of Theorem 1}\label{appxB}
Using Definition \ref{apx:def:avgprojoperator} and \eqref{eq:KA:rowpdf}, the average projection operator of $\mathbf{B^{\scriptscriptstyle \mathrm{H}}}$ is \cite{Boroujerdi2018b}
\begin{equation}
\bar{\mathcal{P}}=\sum_{z=1}^{m}\dfrac{\lVert\mathbf{b}^{\scriptscriptstyle \mathrm{H}}_{r(t)}\rVert^{2}_{2}}{\lVert\mathbf{B}\rVert^{2}_{\mathrm{F}}}\dfrac{\mathbf{b}^{\scriptscriptstyle \mathrm{H}}_{r(t)}\mathbf{b}_{r(t)}}{\lVert\mathbf{b}_{r(t)}\rVert^{2}_{2}}=\dfrac{\mathbf{B}\mathbf{B}^{\scriptscriptstyle \mathrm{H}}}{\lVert\mathbf{B}\rVert^{2}_{\mathrm{F}}}.
\label{apx:proof:th1:avgprojoperator}
\end{equation}
From Definition \ref{apx:def:avggain} and \eqref{apx:proof:th1:avgprojoperator}, the average gain provided by $\mathbf{B}^{\scriptscriptstyle \mathrm{H}}$ can be obtained as \cite{Boroujerdi2018b}
\begin{equation}
{\kappa_{\mathcal{X}}(\mathbf{B}^{\scriptscriptstyle \mathrm{H}})=\min_{\boldsymbol{\vartheta}\in\mathcal{X},\boldsymbol{\vartheta}\neq{0}}\dfrac{\lVert\mathbf{B}^{\scriptscriptstyle \mathrm{H}}\boldsymbol{\vartheta}\rVert^{2}_{2}}{\lVert\mathbf{B}^{\scriptscriptstyle \mathrm{H}}\rVert^{2}_{\mathrm{F}}\lVert\boldsymbol{\vartheta}\rVert^{2}_{2}}\overset{(a)}{=}\min_{\boldsymbol{\varepsilon}\in\mathbb{C}^{K},\boldsymbol{\varepsilon}\neq 0}\dfrac{\lVert\mathbf{B}^{\scriptscriptstyle \mathrm{H}}\mathbf{B}\boldsymbol{\varepsilon}\rVert^{2}_{2}}{\lVert\mathbf{B}^{\scriptscriptstyle \mathrm{H}}\rVert^{2}_{\mathrm{F}}\lVert\mathbf{B}\boldsymbol{\varepsilon}\rVert^{2}_{2}}
	=\dfrac{\lambda_{\min}\left(\mathbf{B}^{\scriptscriptstyle \mathrm{H}}\mathbf{B}\right)}{\lVert\mathbf{B}\rVert^{2}_{\mathrm{F}}},}
\end{equation}
wherein ($a$) it was used the fact that any vector in $\mathcal{X}\subset\mathbb{C}^{M+K}$, which is the subspace generated by the columns of $\mathbf{B}^{\scriptscriptstyle \mathrm{H}}$, can be written as $\mathbf{B}\boldsymbol{\varepsilon}$ with $\boldsymbol{\varepsilon}\in\mathbb{C}^{K}$, as made in \cite{Boroujerdi2018b}.
Finally, the average error of the first, deterministic iteration of Algorithm \ref{algthm:parlrka} can be acquired through Definition \ref{apx:def:projoperator}: $\mathcal{P}_{k}={(\mathbf{b}^{\scriptscriptstyle \mathrm{H}}_{k}\mathbf{b}_{k})} / {\lVert\mathbf{b}_{k}\rVert^{2}_{2}}.$ As $k$ is deterministic, $\bar{\mathcal{P}}=\mathcal{P}_{k}$ and then, completing the proof, the average gain provided by the first iteration is
\begin{equation}
\kappa^{1}_{\mathcal{X}}(\mathbf{B}^{\scriptscriptstyle \mathrm{H}})=\min_{\boldsymbol{\vartheta}\in\mathcal{X},\boldsymbol{\vartheta}\neq0}\dfrac{\boldsymbol{\vartheta}^{\scriptscriptstyle \mathrm{H}}\bar{\mathcal{P}}\boldsymbol{\vartheta}}{\lVert\boldsymbol{\vartheta}\rVert^{2}_{2}}=\min_{\boldsymbol{\vartheta}\in\mathcal{X},\boldsymbol{\vartheta}\neq0}\dfrac{\lVert\mathbf{b}_{k}^{\scriptscriptstyle \mathrm{H}}\boldsymbol{\vartheta}\rVert^{2}_{2}}{\lVert\mathbf{b}_{k}\rVert^{2}_{2}\lVert\boldsymbol{\vartheta}\rVert^{2}_{2}}.
\end{equation}



\end{document}